\documentclass[12pt]{article}
\usepackage{amsmath,amsfonts,amsthm,amssymb}
\usepackage{url}
\usepackage{graphicx}
\usepackage[font=small]{caption}
\usepackage{subcaption}
\usepackage[colorlinks=true,allcolors=blue]{hyperref}
\usepackage{algpseudocode}
\usepackage{algorithmicx}
\usepackage{bbm}
\usepackage{xfrac}
\usepackage{float}
\usepackage{color}
\usepackage{mathtools}
\usepackage{comment}
\usepackage[a4paper, total={6.8in, 9in}]{geometry}
\usepackage{thm-restate}

\newtheorem{theorem}{Theorem}
\newtheorem{lemma}[theorem]{Lemma}

\newcommand{\be}[0]{\begin{enumerate}}
\newcommand{\ee}[0]{\end{enumerate}}
\newcommand{\bi}[0]{\begin{itemize}}
\newcommand{\ei}[0]{\end{itemize}}

\newcommand{\mass}{\textit{mass}}
\newcommand{\size}{\textit{size}}

\title{Analysis of Smooth Heaps and Slim Heaps}

\author{Maria Hartmann\thanks{Institut f\"ur Informatik, Freie Universit\"at Berlin. Work supported by DFG grant KO 6140/1-1.} ~~ L{\'a}szl{\'o} Kozma\thanks{Institut f\"ur Informatik, Freie Universit\"at Berlin. Work supported by DFG grant KO 6140/1-1.} ~~ Corwin Sinnamon\thanks{Department of Computer Science, Princeton University.} ~~ Robert E.\ Tarjan\thanks{Department of Computer Science, Princeton University; and Intertrust Technologies. Research at Princeton University partially supported by an innovation research grant from Princeton and a gift from Microsoft.}}
\date{}

\begin{document}

\maketitle

\abstract{
The smooth heap is a recently introduced self-adjusting heap [Kozma, Saranurak, 2018] similar to the pairing heap [Fredman, Sedgewick, Sleator, Tarjan, 1986]. The smooth heap was obtained as a heap-counterpart of Greedy BST, a binary search tree updating strategy conjectured to be \emph{instance-optimal} [Lucas, 1988], [Munro, 2000]. Several adaptive properties of smooth heaps follow from this connection; moreover, the smooth heap itself has been conjectured to be instance-optimal within a certain class of heaps. Nevertheless, no general analysis of smooth heaps has
existed until now, the only previous analysis showing that, when used in \emph{sorting mode} ($n$ insertions followed by $n$ delete-min operations), smooth heaps sort $n$ numbers in $O(n\lg n)$ time.

In this paper we describe a simpler variant of the smooth heap we call the \emph{slim heap}. We give a new, self-contained analysis of smooth heaps and slim heaps in unrestricted operation, obtaining amortized bounds that match the best bounds known for self-adjusting heaps. Previous experimental work has found the pairing heap to dominate other data structures in this class in various settings. Our tests show that smooth heaps and slim heaps are competitive with pairing heaps, outperforming them in some cases, while being comparably easy to implement.
}

\section{Introduction}

A heap (priority queue) data structure stores a collection of items, each with an associated real-valued key (priority). Operations include inserting an item, deleting an item with smallest key, decreasing the key of an item, or melding (merging) two heaps. Heaps are among the most basic and well-studied structures in computer science, with applications in sorting, event simulation, graph algorithms, and many other settings (see e.g.~\cite{BrodalSurvey, LarkinSenTarjan} and the references therein).

Williams' implicit binary heap~\cite{Williams64} and its variants remain the simplest, and in many applications, the fastest implementation. %
Carefully engineered \emph{sequence-based} heaps~\cite{Sanders00} tend to be even faster in several (although not all) practical scenarios, because of their better use of caching. %

To implement efficient \emph{meld} and to obtain theoretically optimal running times for \emph{insert} and \emph{decrease-key}, more sophisticated data structures have been proposed. The Fibonacci heap~\cite{Fibonacci} was the first heap implementation to obtain the asymptotically optimal amortized bounds of $O(\lg{n})$ time for \emph{delete-min} and \emph{delete} in a size-$n$ heap, and $O(1)$ time for all other operations. (We denote by $\lg$ the base-two logarithm.) A large number of alternative designs that match these bounds have been explored throughout the years, see e.g.~\cite{HollowHeaps} and the references therein. Due to their relatively complex implementation, Fibonacci heaps and other theoretically optimal heaps have been found to be slower in practice than simpler heaps with sub-optimal guarantees~\cite{LarkinSenTarjan}.  

\subparagraph*{Self-adjusting heaps.} Is it possible to design data structures that are simple and efficient in practice but also efficient in theory, at least in an amortized sense? Self-adjusting data structures attempt to achieve this goal by allowing a very flexible structure with no (or very little) bookkeeping. They react to user operations with local structural readjustments, attempting to bring the data structure into a more favorable state for future operations.  Examples include the \emph{splay tree}~\cite{ST85}, a self-adjusting binary search tree, and the \emph{pairing heap}~\cite{FSST86}, a self-adjusting heap. 

A pairing heap is a rooted, ordered tree (the children of each node are ordered) whose nodes are the heap items, arranged in (\emph{min}-)\emph{heap} order: the parent of a node $x$ has key no greater than that of $x$.  Heap order implies that the root has minimum key.  A \emph{delete-min} operation deletes the root, thereby making each of its children into a root.  These new roots retain the order they had when they were children of the deleted root.  The \emph{delete-min} combines the new roots into a single tree by \emph{linking} them two at a time.  One link operation combines two adjacent roots, making the one with larger key the new \emph{leftmost} child of the other.
In its original variant, the pairing heap combines roots by first linking pairs of adjacent roots (the first and second, the third and fourth, etc.) from left to right, and then linking the remaining roots consecutively from right to left. 
This variant performs all operations in $O(\lg n)$ amortized time~\cite{FSST86}.   Although it was originally conjectured that pairing heaps match the theoretical guarantees of Fibonacci-heaps, Fredman~\cite{FredmanLB} and Iacono and {\"O}zkan~\cite{IaconoOzkan} showed that the amortized time of \emph{decrease-key} in pairing heaps is $\Omega{(\lg\lg{n})}$, and that the same lower bound holds for broader classes of self-adjusting heaps. %

The exact amortized complexity of the \emph{decrease-key} operation in pairing heaps remains unknown~\cite{IaconoPairing, PettiePairing, DKKPZ}\footnote{The lack of structure of self-adjusting data structures makes their analysis both difficult and interesting.}. In practice, however, pairing heaps have been found to dominate Fibonacci heaps and other theoretically optimal heaps across a wide range of experimental settings, and for some specific types of workloads (e.g.\ with a high frequency of \emph{decrease-key} operations) they even compare favorably with implicit heaps~\cite{LarkinSenTarjan}. The pairing heap can thus be considered a ``robust choice''~\cite{SMDD19} in a variety of applications.

In addition, pairing heaps and other self-adjusting data structures hold the promise of \emph{adaptivity}, i.e.\ the ability to take advantage of regularities or biases in the usage pattern, which may give improved performance in specific scenarios. 
The adaptivity of splay trees 
has inspired a long line of research and is the subject of the \emph{dynamic optimality conjecture}~\cite{ST85, SplayNew}. Adaptivity in heaps is relatively less studied (see e.g.~\cite{KS19} and references therein).

\subparagraph*{Smooth heaps and slim heaps.} 

The recently introduced \emph{smooth heap}~\cite{KS19} is a self-adjusting heap with a structure similar to that of the pairing heap. Like pairing heaps, smooth heaps are built of rooted, ordered trees combined by links. The two data structures differ however, both in the order in which they do links, and in the implementation of the linking primitive itself. Smooth heaps use \emph{stable linking}: when linking two adjacent roots, the one with larger key becomes the leftmost---or the \emph{rightmost}---child of the other, depending on the original left-to-right order of siblings. %

The smooth heap was obtained from a correspondence between heaps and binary search trees. In sorting mode ($n$ insertions followed by $n$ delete-min operations), the smooth heap has the same asymptotic running time as Greedy BST, a self-adjusting binary search tree strategy\footnote{The correspondence is subtle: going from the smooth heap to Greedy BST involves \emph{inverting} and \emph{reversing} the permutation that is being sorted.}. Greedy BST was proposed as a candidate  \emph{dynamically optimal} search tree implementation~\cite{Luc88, Mun00, DHIKP09}. In its original form, Greedy BST needs information about \emph{future} queries. Remarkably, this requirement can be removed (with a constant factor slowdown)~\cite{DHIKP09}.  Nonetheless, Greedy BST should be seen as a theoretical ``proof of concept'' that is largely impractical, as compared to splay trees, which are simple and widely used in practice. Surprisingly, despite its correspondence to Greedy BST, the smooth heap is simple and easy to implement.  In particular, operations do not require knowledge of the future. The pointer structure and low-level complexity of operations in smooth heaps are comparable to those in pairing heaps. 

In~\cite{KS19}, several bounds were shown for the time of operations on smooth heaps used in sorting mode, by transferring results known for Greedy BST. %
The correspondence with Greedy BST breaks down when additional operations are supported, or even when \emph{insert} and \emph{delete-min} operations are intermixed. The smooth heap may be a viable general-purpose data structure, but so far it has lacked a complete, self-contained analysis that covers unrestricted operations.
 
In this paper we provide such an analysis.  We begin by describing the implementation of all standard heap operations, some of which were not fully specified in~\cite{KS19}.  We also describe and analyze the variant of smooth heaps that uses unstable (classical) linking, which we call the \emph{slim heap}.  The bounds we obtain for smooth heaps and slim heaps match the best known bounds for any self-adjusting heap, and differ from those of Fibonacci heaps only in the $O(\lg\lg{n})$ bound (versus $O(1)$) for \emph{decrease-key} (see Figure~\ref{fig:times}). Despite the similarities between pairing heaps, smooth heaps, and slim heaps, the analysis of pairing heaps does not directly transfer to smooth heaps or slim heaps, since the behaviours of the data structures differ. Indeed, even minor variants of the standard pairing heap are difficult to analyze, with open questions remaining.  See e.g.\ \cite{DKKPZ}.  On the other hand, our analysis of smooth heaps and slim heaps does use a variant of the potential function used to analyze pairing heaps.  

Our $O(\lg\lg n)$ bound for \emph{decrease-key} in smooth heaps and slim heaps matches the lower bounds of Fredman~\cite{FredmanLB} and Iacono and \"Ozkan~\cite{IaconoOzkan}, but our implementation, which is based on Elmasry's~\cite{Elmasry17} implementation for pairing heaps, violates the assumptions of these bounds, so they do not in fact apply to our algorithm. It remains open whether these lower bounds can be extended to encompass an Elmasry-type implementation.  A related question is whether there is \emph{any} self-adjusting heap with an $O(\lg\lg{n})$ bound for \emph{decrease-key} that satisfies the assumptions of these lower bounds.  The \emph{sort heap}~\cite{IaconoOzkan} satisfies the assumptions of the Iacono-\"{O}zkan bound, but its analysis is flawed~\cite{IStalk}, \cite[\S\,6]{Hartmann_thesis}.  In this paper, we give a simple implementation of \emph{decrease-key} in smooth heaps and slim heaps based on the original implementation for pairing heaps. This implementation does satisfy the assumptions of Iacono and \"Ozkan's lower bound, and two authors of this paper have shown that the simple decrease-key also takes $O(\lg\lg n)$ time for both smooth and slim heaps~\cite{STdecrease}, with all other operations unchanged, matching the lower bound to within a constant factor.

We complement our theoretical results with a brief experimental study (Section~\ref{sec:exp}). Since the pairing heap was previously found to have good experimental performance and has been extensively compared with other heaps, we limit ourselves to comparing smooth heaps and slim heaps with pairing heaps.   
Our initial results show that smooth heaps and slim heaps are competitive with pairing heaps, outperforming pairing heaps in some cases, particularly for some structured usage patterns.  This aligns with the adaptive properties shown previously in an asymptotic sense for smooth heaps.  Our results suggest the smooth heap or slim heap as a drop-in alternative heap in applications where the pairing heap works well.

\section{Smooth heaps and slim heaps}\label{sec:smooth}

In this section, we introduce the smooth heap and its variant, the slim heap, and state our efficiency bounds.
Both of these data structures store a collection of nodes, each node having a real-valued key (not necessarily distinct for different nodes).
They support the following standard heap operations:
\bi
\item \emph{make-heap}($h$): Create a new, empty heap $h$.
\item \emph{find-min}($h$): Return a node of smallest key in heap $h$, or null if $h$ is empty.
\item \emph{insert}($x, h$): Insert node $x$ with predefined key into heap $h$.  Node $x$ must be in no other heap.
\item \emph{delete-min}($h$): Delete from $h$ the node that would be returned by \emph{find-min}($h$).
\item \emph{meld}($h, h'$): Meld node-disjoint heaps $h$ and $h'$.
\item \emph{decrease-key}($x, k, h$): Decrease to $k$ the key of node $x$ in heap $h$, assuming that $x$ is a node in $h$ and $k$ is no larger than the current key of $x$.
(Operation \emph{decrease-key} is given a pointer to node $x$ in heap $h$, not just its key or some other identifier.)
\item \emph{delete}($x, h$): Delete node $x$ from heap $h$, assuming that $x$ is a node in heap $h$. Again, \emph{delete} is given a pointer to $x$ in $h$.
\ei

\begin{figure}[!h]
\begin{center}
\begin{tabular}{| c || c | c | c | c |}
\hline
& \emph{insert} & \emph{delete-min} & \emph{meld} & \emph{decrease-key}\\
\hline
\textbf{Smooth Heaps} (this paper) & $O(1)$ & $O(\lg n)$ & $O(1)$ & $O(\lg\lg{n})$\\
\hline
\textbf{Slim Heaps} (this paper) & $O(1)$ & $O(\lg n)$ & $O(1)$ & $O(\lg\lg{n})$\\
\hline
Fibonacci Heaps~\cite{Fibonacci} & $O(1)$ & $O(\lg n)$ & $O(1)$ & $O(1)$\\
\hline
Pairing Heaps~\cite{FSST86, StaskoVitter, IaconoPairing} & $O(1)$ & $O(\lg n)$ & $O(1)$ & $O(\lg n)$ \\
Alternative Analysis (Pettie~\cite{PettiePairing}) & $O(4^{\sqrt{\lg\lg{n}}})$ & $O(\lg n)$ & $O(4^{\sqrt{\lg\lg n}})$ & $O(4^{\sqrt{\lg\lg n}})$\\
Elmasry Pairing Heaps~\cite{Elmasry17}\footnotemark & $O(1)$ & $O(\lg n)$ & $O(\lg\lg n)$ & $O(\lg\lg{n})$\\
\hline
\end{tabular}
\end{center}
\caption{The best time bounds known for smooth, slim, Fibonacci, and pairing heaps.
All bounds are amortized. Smooth heaps are competitive with the best variants of pairing heaps.
They fall short of Fibonacci heaps on \emph{decrease-key} (though self-adjusting heaps often perform better in practice~\cite{LarkinSenTarjan}).
In all cases, \emph{make-heap} and \emph{find-min} take $O(1)$ time, and \emph{delete} takes $O(\lg n)$ time.
Strict Fibonacci heaps~\cite{BLT12} achieve the same bounds as Fibonacci heaps, but without amortization.}
\label{fig:times}
\end{figure}
\footnotetext{Elmasry~\cite{Elmasry17} also claims a \emph{meld} time of $O(1)$ with a \emph{delete-min} time of $O(\lg n + \lg\lg{N})$, where $N$ is the number of items in all heaps.}

\subsection{Efficiency}

Figure~\ref{fig:times} states our time bounds for smooth heaps and summarizes known results for some competing heaps.  In stating these bounds, we assume that $n \geq 4$, so that $\lg\lg{n} \geq 1$: if $n < 4$ all operations take $O(1)$ time.  We shall prove the following theorem:

\begin{theorem}
\label{thm:unweighted}
In a smooth heap or slim heap, \emph{delete-min} and \emph{delete} take $O(\lg n)$ amortized time, \emph{decrease-key} takes $O(\lg\lg{n})$ amortized time, and \emph{make-heap}, \emph{find-min}, \emph{insert}, and \emph{meld} take $O(1)$ time, where there are $n$ nodes in the heap at the time of operation.
\end{theorem}

In some applications, a large number of nodes may be inserted and then never accessed, for example if one inserts many nodes with large keys, and then repeatedly inserts and deletes a few
nodes with small keys. For this case, we shall prove the following theorem, which states that nodes that are never deleted do not slow down deletions:

\begin{restatable}{theorem}{restatethma}\label{thm:remain}
Operations \emph{delete-min} and \emph{delete} take $O(\lg t)$ amortized time, with the other amortized bounds unchanged, where $t$ is the number of nodes in the heap at the time of the operation that will be deleted in the future.
\end{restatable}

\subsection{Data structures and terminology}
\label{sec:structure}

A smooth heap or slim heap consists of a forest of rooted, ordered trees whose nodes are the nodes of the heap.  Each tree is (min-)heap ordered: the key of a non-root node is no less than that of its parent.
The roots of the trees are stored in a list, with a root of minimum key, the \emph{min-root}, first.  Each node stores a list of its children.
For both the root list and the lists of children, we identify the front of the list with ``left'' and the back of the list with ``right'', so that nodes early in the list are considered left of nodes later in the list.

The forest is altered by \textit{linking} pairs of \emph{adjacent} nodes in the list of roots or in a list of children.  A link makes the node with smaller key (the \emph{winner} of the link) the parent of the node of larger key (the \emph{loser} of the link).  If the keys are equal, the node on the right is the winner.  The link is a \emph{left link} if the loser is originally left of the winner, a \emph{right link} if the loser is originally right of the winner. 

The only difference between smooth and slim heaps is the position of the loser in its list of new siblings after a link. (See Figure~\ref{fig:link-types}.)
Slim heaps use \textit{one-sided} links: the loser of the link becomes the new leftmost child of the winner.  This is the type of linking used in Fibonacci heaps, pairing heaps, and many other similar data structures.

\begin{figure}[!h]
\begin{center}
\includegraphics[scale=0.3]{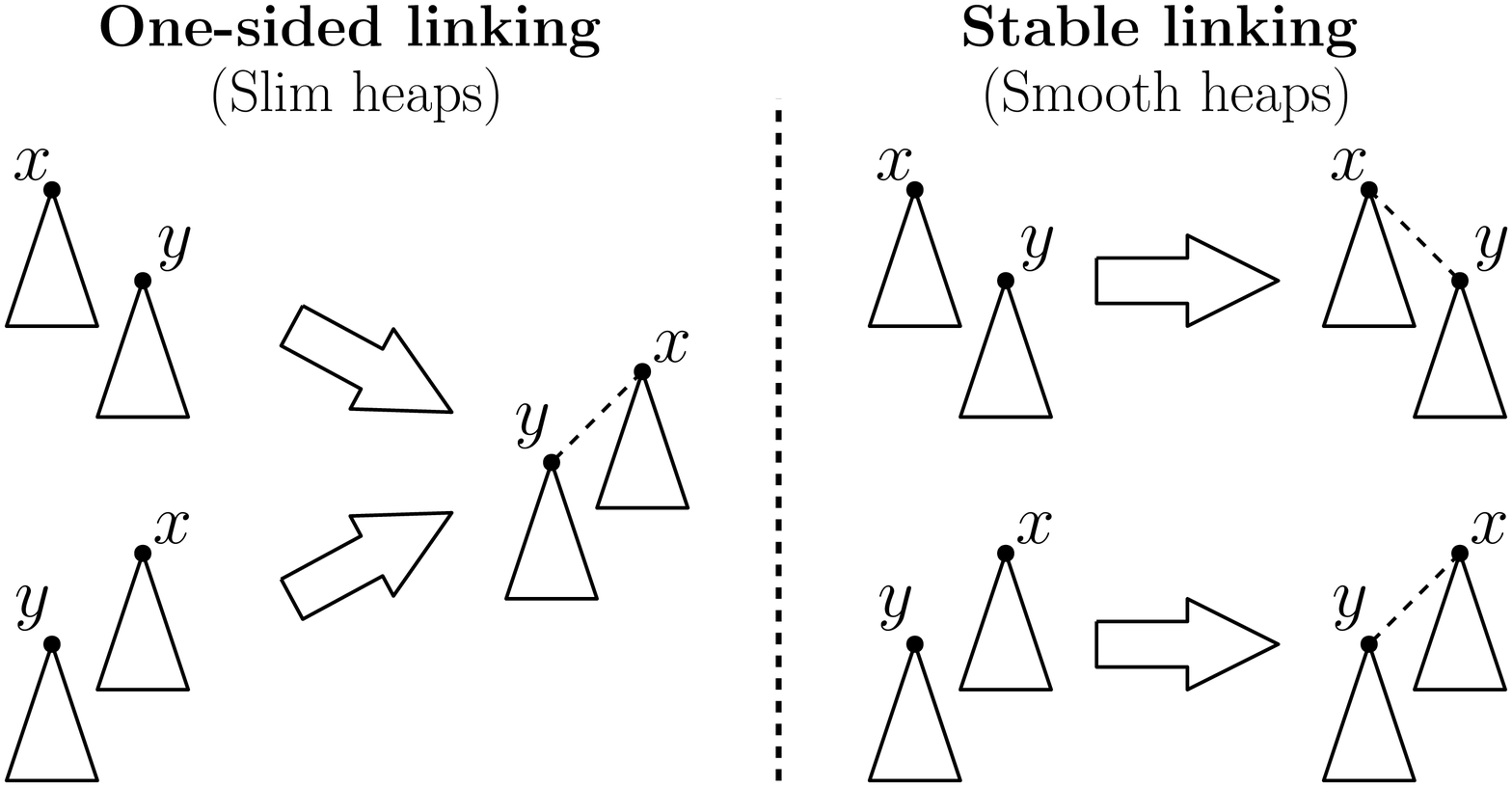}
\caption{The result of a link of adjacent nodes $x$ and $y$, where $x$ has the smaller key, using one-sided linking (left) and stable linking (right). A one-sided link always adds the child on the left, whereas a stable link preserves left-to-right order. Slim heaps use one-sided links, smooth heaps use stable links. In both smooth and slim heaps, all links are between adjacent nodes.}
\label{fig:link-types}
\end{center}
\end{figure}

Smooth heaps use \textit{stable} links: the loser becomes the new leftmost or rightmost child of the winner, depending on whether the link is a left link or a right link.  Stable links maintain the left-to-right order of nodes (although \emph{insert}, \emph{meld}, and \emph{decrease-key} operations perturb it, as we shall see).
A child of a node is either a \textit{left} child or a \textit{right} child, depending on whether it was left or right of its parent just before they were linked.  In a smooth heap all left children in a list of children precede all right children; in a slim heap, left and right children are in general intermixed.

\subsection{Operations}\label{sec:operations}

In this subsection we describe the implementation of all the heap operations except \emph{decrease-key} and \emph{delete}; we cover these in Section~\ref{sec:decrease-key}.  The description of smooth heaps in~\cite{KS19} did not include a find-min operation, nor was the implementation of \emph{decrease-key} fully specified. In addition to adding these operations, we have changed some details of the original presentation.

As noted in the last subsection, slim heaps differ from smooth heaps only in the linking method. The following descriptions apply to both smooth heaps and slim heaps if links are done using the appropriate linking method, stable for smooth heaps, one-sided for slim heaps.

We store the root list of a heap in a circular singly-linked list, with the min-root first (leftmost).  Access is via the min-root.  Circular linking supports melding in $O(1)$ time.  We store each list of children in a singly-linked list, circular for smooth heaps.  In a slim heap, access to a list of children is via the leftmost child.  In a smooth heap, access to such a list is via the \emph{rightmost} child; circular linking allows a new child to be added as the leftmost or rightmost in $O(1)$ time.  Each node needs two pointers: to its right adjacent root or sibling, and to its leftmost (slim heap) or rightmost (smooth heap) child.  To support \emph{delete} and \emph{decrease-key} we shall add a third pointer per node.  See Section~\ref{sec:decrease-key}.

\bi
\item \emph{make-heap}($h$): Create and return an empty heap.

\item \emph{find-min}($h$): Return the min-root of $h$.

\item \emph{insert}($x, h$): Make $x$ into a one-node tree and insert $x$ into the root list of $h$, in first or second position depending on whether its key is less than or not less than that of the min-root of $h$; in the former case $x$ becomes the new min-root.

\item \emph{meld}($h, h'$): Catenate the root lists of $h$ and $h'$; set the min-root of the melded heap to that of $h$ or $h'$, whichever has smaller key.

\item \emph{delete-min}($h$): Delete the min-root $x$ of $h$.  Replace $x$ in the list of roots of $h$ by the list of children of $x$, in the order they occur in the list of children.  These \emph{new roots} precede all the undeleted (\emph{old}) roots in the root list.  Repeatedly link pairs of adjacent roots until there is only one root remaining, using the \emph{leftmost locally maximum linking rule} given below.  Make the remaining root the min-root of $h$.  

\ei

The main novelty in smooth heaps and slim heaps is the leftmost locally maximum linking rule used in \emph{delete-min}.  It is: find the leftmost node $v$ in the root list whose key is no less than those of both adjacent nodes, say $u$ and $w$, and link $v$ with whichever of $u$ and $w$ has larger key, breaking a tie in favor of the node left of $v$.  As special cases, if the leftmost root has key no less than that of the right adjacent root, link these roots (the right one is the winner); if the keys of all the roots are in strictly increasing order left-to-right, link the rightmost root with the left adjacent root (the left one is the winner).  One can eliminate the special cases by adding dummy leftmost and rightmost nodes, both with key minus infinity, and applying the leftmost locally maximum linking rule until there is only one non-dummy root.)

To implement leftmost locally maximum linking, start at the leftmost root and proceed rightward until finding a root $v$ whose key is not less than that of its right adjacent root $w$, or until reaching the rightmost root.  In the latter case, link the two rightmost roots and repeat until there is only one root, completing the linking.  In the former case, link $v$ with whichever of the adjacent roots has higher key, choosing the left adjacent node in case of a tie and linking with $w$ (the right adjacent node) if $v$ is the leftmost root.  Then repeat, but starting from the winner of the link: all roots preceding this winner are in strictly decreasing order by key, left to right, and hence none is the leftmost local maximum.

Leftmost locally maximum linking \emph{treapifies} the list of roots: it arranges the roots into a binary tree symmetrically ordered by root list order and heap-ordered by key, with ties broken in favor of nodes on the right: if $x$ is a node, its new left child, if any, has no smaller key, and its new right child, if any, has strictly greater key.

In both linking and in the definition of ``leftmost local maximum'' our tie-breaking rule in key comparisons is that the node on the right is treated as having smaller key.  Other tie-breaking rules work equally well: what is required is consistency (the outcome of a key comparison between two nodes must remain the same throughout a \emph{delete-min}) and transitivity (if $x$, $y$, and $z$ have equal keys and the tie-breaking rule declares that $x$ has smaller key than $y$ and $y$ has smaller key than $z$, then it must declare that $x$ has smaller key than $z$).  An alternative tie-breaking rule with these properties is to break ties by node identifier.  The tie-breaking rule can depend on the keys, the nodes, and the position of the nodes in the root list.  Leftmost locally maximum linking is easy to implement, but a more general linking rule does the same links if the tie-breaking rule is the same: Find \emph{any} locally maximum root $v$ (a root with key greater than those of both adjacent roots); link it with the adjacent node of greater key; repeat until one root remains.

Smooth heaps and slim heaps are efficient for two reasons.  The first is that locally maximum linking guarantees that a node acquires at most two new children during a \emph{delete-min} operation (Lemma~\ref{lem:binary}).   Pairing heaps do not have this property: in the second, right-to-left linking pass during a pairing heap \emph{delete-min}, a node can acquire an arbitrary number of new children (see Figure~\ref{fig:stable-link}).  The second is that insertions and melds are lazy (they do no links), but \emph{delete-min} operations are eager (they do as many links as possible).  Pairing heaps are eager: except in the middle of an operation, a pairing heap consists of a single tree, and insertions and melds are done by single links.  We do not know whether it is possible to obtain our bounds for a single-tree smooth heap or pairing heap.   

\begin{figure}[!ht]
\begin{center}
\includegraphics[scale=0.35]{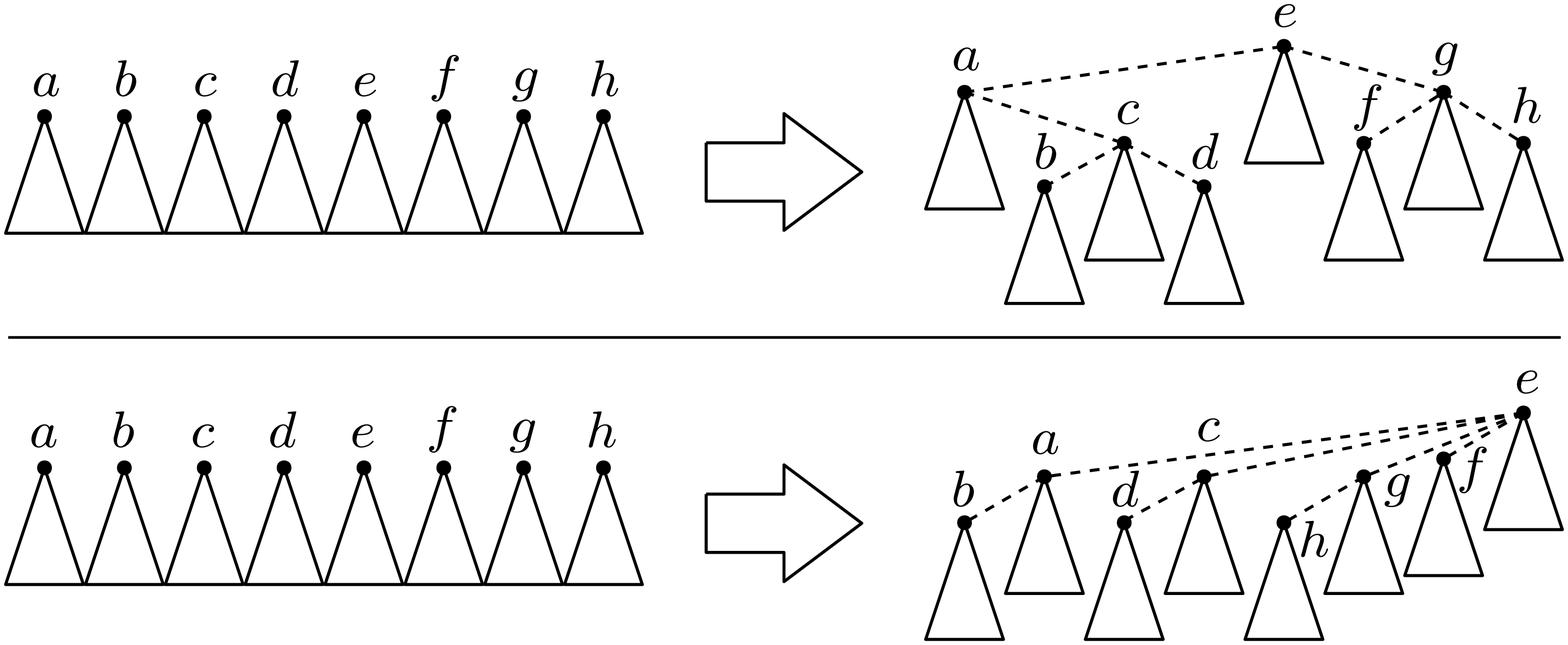}
\caption{A possible linking of roots during a \emph{delete-min} in a smooth heap (above) and a pairing heap (below), assuming key order $e, a, g, c, b, d, f, h$ from smallest to largest. }
\label{fig:stable-link}
\end{center}
\end{figure}

\begin{lemma}\label{lem:binary}
During a \emph{delete-min}, a root wins at most one left link and at most one right link.  Hence during a \emph{delete-min} each node acquires at most one new left child and at most one new right child. 
\end{lemma}

\begin{proof}
Consider a link of the leftmost locally maximum node $v$ during a \emph{delete-min}.  Node $v$ is linked with either the left adjacent node $u$ or the right adjacent node $w$.  Suppose the link is with $u$. Then $u$ has key strictly less than that of $v$; otherwise, $u$ has locally maximum key and is left of $v$, contradicting the choice of $v$. The link is a right link won by $u$.  After this link, either $u$ is the rightmost node in the list or the node $w$ adjacent to $u$ on the right has no larger key, and the node adjacent to $u$ on the left, if any, has key strictly smaller than that of $u$.  This makes $u$ the leftmost local maximum, so $u$ will participate in the next link and lose it: if the link is with $w$, $w$ will win by the link tie-breaking rule if the keys of $u$ and $w$ are equal. Hence after the link of $u$ and $v$, $u$ cannot win another link, either left or right. 

Suppose the link of $v$ is with $w$.  Since the key of $w$ is no greater than that of $v$, $w$ wins the link, by the link tie-breaking rule if the keys of $v$ and $w$ are equal.  The link is a left link.  After the link, the node $u$ adjacent to $w$ on the left, if it exists, has key strictly smaller than that of $w$ by the tie-breaking rule for choosing whether to link $v$ with $u$ or with $w$. Node $u$ cannot become a leftmost local maximum until $w$ loses a link.  Hence $w$ cannot win another left link, since it would be with $u$.
\end{proof}

\section{Efficiency of operations}
\label{sec:analysis}

With the two-pointer-per-node representation in Section~\ref{sec:operations}, each \emph{make-heap}, \emph{find-min}, \emph{insert}, and \emph{meld} operation takes $O(1)$ time worst-case: catenation of the two root lists during \emph{meld} takes two pointer changes, and each link, whether one-sided or stable, takes O(1) time.  The time for a \emph{delete-min} is O(1) plus at most a constant times the number of links.  We normalize this time to be one plus the number of links. The number of links is one less than the number of roots, new and old, after deletion of the min-root.  We shall prove a bound of $O(\lg n)$ on the amortized time of \emph{delete-min} and of $O(1)$ on the amortized times of the other operations by using the \emph{potential method}~\cite{tarjan1985amortized}.   

The previous analysis of smooth heaps~\cite{KS19} was restricted to \emph{sorting mode}, and made use of a \emph{geometric view} of binary search trees~\cite{DHIKP09}. Our analysis is free of such restrictions and is self-contained.

\subsection{Potential method}

In the potential method, we assign to each state of the data structure a real-valued potential.  We define the \emph{amortized time} of an operation to be its actual time plus the potential of the structure after the operation minus the potential of the structure before it.  That is, the amortized time is the actual time plus the net increase in potential caused by the operation.  If we sum the amortized times of the operations in a sequence, the sum of the potential differences telescopes: the sum of the actual times of the operations equals the sum of their amortized times, plus the final potential (after the last operation), minus the initial potential (before the first operation).  If the initial potential is $0$ (corresponding to an empty data structure) and the final potential is non-negative, then the sum of the amortized times of the operations is an upper bound on the sum of their actual times, allowing us to use the former as a conservative bound on the latter.

\subsection{Definition of the potential}\label{sec-dfn-potential}

The analyses of slim heaps and smooth heaps differ slightly.  We develop them both at the same time and point out the differences.
    
For each node $x$ having children, we define the \emph{link order} of the children to be the order in which these children lost links to $x$, latest first, earliest last.  This is exactly the order of the children in the list of children of $x$ if the heap is a slim heap, but not necessarily if the heap is a smooth heap: in the latter, the link order is a merge of the left children in their order in the list of children of $x$ and of the right children in the reverse  of their order in the list of children of $x$.

We define the size $\size(x)$ of a node $x$ in a heap to be the number of its descendants, including itself. We define the {mass} $\mass(x)$ of a child $x$ to be the sum of the sizes of $x$ and of all its siblings after it in link order (those linked to their common parent before $x$).  (We do not need to define the masses of roots.)

We define the potential of a root $x$ to be $2+2\lg \size(x)$.  In a slim heap, we define the potential of a child $x$ to be $0$ if $x$ is the first or second child of its parent, $\lg \mass(x)$ otherwise (at least two children were linked to the parent of $x$ after $x$).  In a smooth heap, we define the potential of a child $x$ to be $0$ if it is the leftmost left child or rightmost right child in its sibling list, $\lg \mass(x)$ otherwise.  We define the potential of a collection of heaps to be the sum of the potentials of their nodes.

The slim heap potential is closely related to the one used in the original analysis of pairing heaps~\cite{FSST86}, which assigns $\lg \size(x)$ potential to each root $x$ and $\lg \mass(x)$ potential to each child $x$.  Here we give (at most) two children of a node zero potential.  This limits the increase in potential caused by a link in a \emph{delete-min} to the logarithm of the size of the winner \emph{before} the \emph{delete-min}.  Our analysis relies on this limit.

\subsection{Amortized bounds}\label{bounds}

The amortized time of a \emph{make-heap}, \emph{find-min}, \emph{insert}, or \emph{meld} is $O(1)$, since the worst-case time of each is $O(1)$, and only an \emph{insert} changes the potential, increasing it by $2$.

We shall show that the amortized time of a \emph{delete-min} is $O(\lg n)$.

Let $x$ be the min-root. After its deletion, the list of roots consists of new roots, those that were children of $x$, followed by old roots, those that were roots before the deletion of $x$.  See Figure~\ref{fig:old-new-roots}.  We separately analyze links won by old roots and links won by new roots.

\begin{figure}[!h]
\centering{
\includegraphics[scale=0.45]{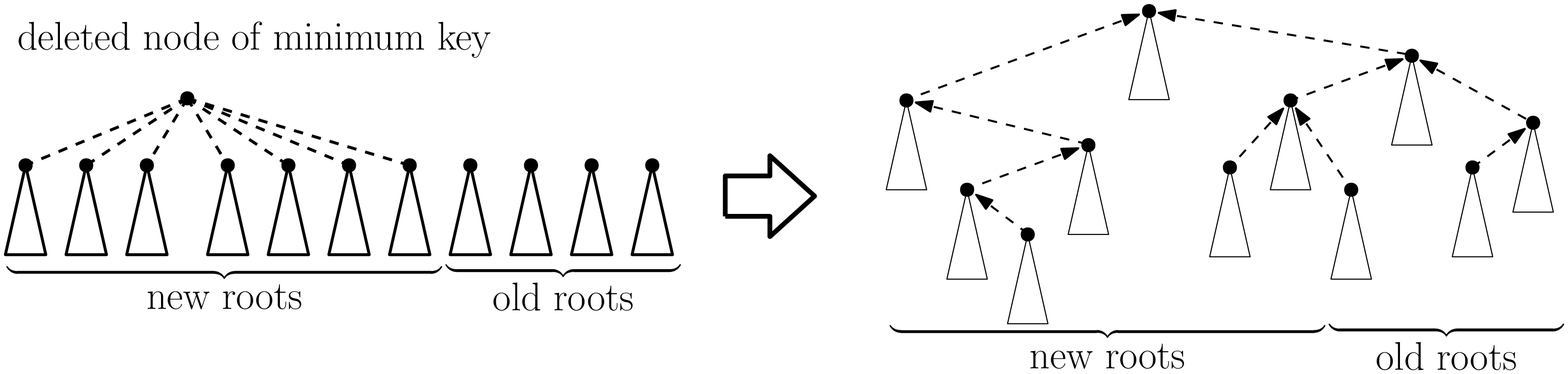}
}
\caption{The result of the delete-min operation. The root with smallest key is deleted, its children become the new roots in its place, and then all roots are linked in a binary tree.}
\label{fig:old-new-roots}
\end{figure}

Each old root $u$ has potential $2+2\lg \size(u)$, where $\size(u)$ is the size of $u$ just before the \emph{delete-min}.  This potential covers the actual time ($1$ or $2$) of the link or links that $u$ wins during the \emph{delete-min}, if there are any, plus the increase in the potential of children of $u$ when $u$ wins a link or two: if $v$ and $w$ respectively are the at most two children of $u$ that have potential zero before the \emph{delete-min}, then when $u$ wins one link, the potential of one of $v$ and $w$ increases from $0$ to $\lg \mass(w) < \lg \size(u)$, and when $u$ wins a second link the potential of the other of $v$ and $w$ increases from $0$ to $\lg \mass(v) < \lg \size(u)$.   

The heart of the analysis, the hard part, is that of links won by new roots.  Let the masses of the new roots be their masses before the deletion of $x$.  Before this deletion, $x$ has potential $2+ 2\lg \size(x)$.  The deletion of $x$ frees its potential.  We transfer $\lg \size(x)$ of this potential from $x$ to each of the at most two new roots with zero potential. Then every new root $u$ has a potential of at least $\lg \mass(u)$.

Each new root $u$ can acquire one or two new children by winning one or two links during the \emph{delete-min}.  By the argument we used for old roots, the increase in potential of the children of $u$ when $u$ wins a link is less than $\lg \size(u)$ per link, where $\size(u)$ is the size of $u$ just before the \emph{delete-min}.  We shall show that the sum of the base-two logarithms of the masses of the new roots, which is at most the sum of their potentials after $x$ is deleted and its freed potential is transferred, plus $2 + \lg \size(x)$ for a slim heap or $2 + 2\lg \size(x)$ for a smooth heap, is at least the sum of the increases in potential of the children of these roots, plus at least $2$ per new root that wins a link.  The extra $2$ pays for the actual time of the links won by new roots during the \emph{delete-min}.

The first step is to shift some of the potential of the new roots so that each new root $u$ has at least $\lg \mass(u)$ potential for each link it wins during the \emph{delete-min}.  This shifting is different for slim heaps and smooth heaps.  In slim heaps, for each new root $u$ except the last one in link order, if the root $v$ after $u$ in link order wins a left link, we shift $\lg \mass(u) \geq \lg \mass(v)$ from $u$ to $v$. See Figure~\ref{fig:shift-slim}.

\begin{figure}[h!]
\centering{
\includegraphics[scale=0.4]{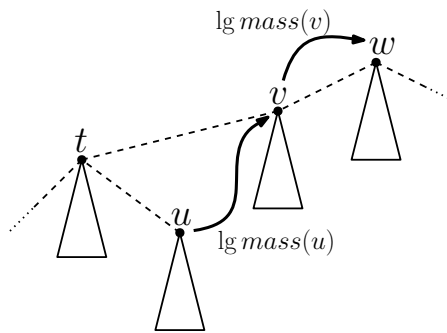}
}
\caption{An example of shifting potential among new roots in a slim heap. The dotted lines represent links made during delete-min, with the losers are pictured lower than the winners, and curved lines represent shifted potential. New root $u$ gives potential to its adjacent sibling $v$ because $v$ wins a left link with $t$ and $v$ gives potential to $w$ because $w$ wins a left link with $v$.}
\label{fig:shift-slim}
\end{figure}

\begin{lemma}\label{L-slim-heap-node-potential}
After the potentials are shifted in a slim heap, each new root $u$ has potential at least $\lg \mass(u)$ if it wins one link during the \emph{delete-min}, at least $2\lg \mass(u)$ if it wins two. 
\end{lemma}

\begin{proof}
Consider a new root $v$ that wins one or two links during the \emph{delete-min}.  If $v$ wins a left link, it acquires at least $\lg \mass(v)$ potential from the root before it in link order.  For $v$ to lose potential to the node $w$ after it in link order, $w$ must win a left link.  But then $v$ cannot win a right link, since if this link occurs before the left link that $w$ wins, $w$ must be the loser, and hence cannot later win a link, and if this link occurs after the left link that $w$ wins, $v$ must be the loser, and cannot later win a link.  We conclude that if $v$ wins both a left link and a right link, it acquires at least $\lg \mass(v)$ potential from the node before it in link order and retains its own potential of at least $\lg \mass(v)$; if $v$ wins only a right link, it retains its own potential of at least $\lg \mass(v)$; and if $v$ wins only a left link, it acquires at least $\lg \mass(v)$ potential from the node before it in link order, although it may lose its own potential.  Hence the lemma holds.
\end{proof}

In the case of smooth heaps, we call a new root a \emph{left} or \emph{right} root depending on whether it was a left or right child of $x$ before $x$ was deleted.  We must shift the potential of left and right roots separately, since the link order of the latter is the reverse of their order on the list of children of $x$ before $x$ was deleted. Specifically, for each left root $u$ except the last in link order, if the next left root $v$ in link order wins a left link, we shift $\lg \mass(u) \geq \lg \mass(v)$ from $u$ to $v$; for each right root $u$ except the last in link order, if the next right root $v$ in link order wins a right link, we shift $\lg \mass(u) \geq \lg \mass(v)$ from $u$ to $v$.  In addition, if the first right root $u$ in link order wins a (right) link with an old root, we give $u$ $\lg \size(u) \leq \lg \size(x)$ of the $2 + 2\lg \size(x)$ extra potential allocated to the \emph{delete-min}.  This accounts for the extra potential needed in the analysis for smooth heaps: since all the new roots precede all the old roots, a new root can only win a right link with an old root, not a left link.  This is a problem only when shifting potential to cover right links, which does not happen in slim heaps. See Figure~\ref{fig:shift-smooth}.

\begin{figure}[h!]
\begin{center}
\includegraphics[scale=0.45]{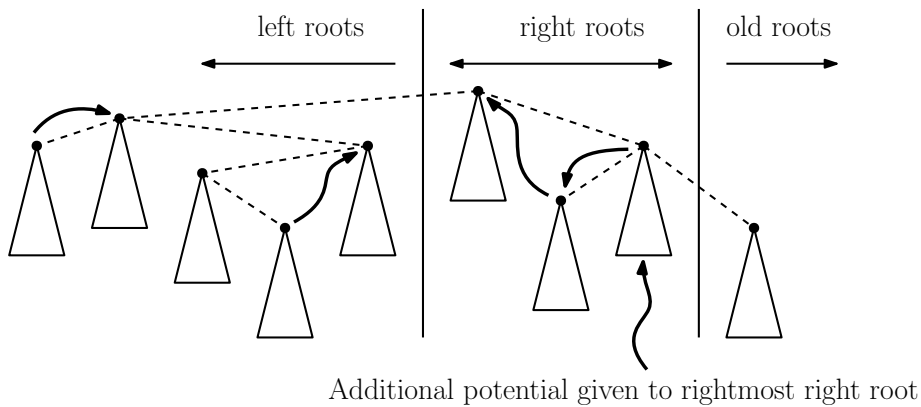}
\end{center}
\caption{An example of shifting potential in a smooth heap. Every curved line represents a transfer of potential equal to the logarithm of the mass of the new root at the start of the line. Left roots transfer potential in the same way as a slim heap. Right roots transfer potential in the symmetric way, and the rightmost right root receives additional potential to pay for a right link.}
\label{fig:shift-smooth}
\end{figure}

\begin{lemma}\label{L-smooth-heap-node-potential}
After the potentials are shifted in a smooth heap, each new root $u$ has potential at least $\lg \mass(u)$ if it wins one link during the \emph{delete-min}, at least $2\lg \mass(u)$ if it wins two. 
\end{lemma}

\begin{proof} The proof is that of Lemma~\ref{L-slim-heap-node-potential} applied separately to left roots and right roots.  Consider a left root $v$ that wins one or two links during the \emph{delete-min}.  If $v$ wins a left link, it acquires at least $\lg \mass(v)$ potential from the left root preceding it in link order: there must be such a root since if $v$ wins a left link it must be with a left root.  For $v$ to lose potential, it must be to the left root $w$ after it in link order, and $w$ must win a left link.  But then $v$ cannot win a right link, since if this link occurs before the left link that $w$ wins, $w$ must be the loser, and hence cannot later win a link, and if this link occurs after the left link that $w$ wins, $v$ must be the loser, and cannot later win a link.  We conclude that if $v$ wins both a left link and a right link, it acquires at least $\lg \mass(v)$ potential from the left root preceding it in link order and retains its own potential of at least $\lg \mass(v)$; if $v$ wins only a right link, it retains its own potential of at least $\lg \mass(v)$; and if $v$ wins only a left link, it acquires at least $\lg \mass(v)$ potential from the left root preceding it in link order, although it may lose its own potential. The symmetric argument applies to right roots, except that the first right root in link order may win a right link (with an old root), and not receive any potential for it.  The additional $\lg \size(x)$ potential allocated to a smooth heap \emph{delete-min} covers this.  Hence the lemma holds.
\end{proof}

By Lemmas \ref{L-slim-heap-node-potential} and \ref{L-smooth-heap-node-potential}, every new root has enough potential to cover the increase in the potential of its children if it wins one or two links.  For each new root $u$ that wins two links during the linking, we allocate $\lg \size(u)$ potential from $u$ to the child whose potential increases when the second link occurs. Each new root $u$ that wins at least one link retains at least $\lg \mass(u)$ potential.  We shall show that if we add $2+\lg \size(x)$ additional potential, there are at least two units of extra potential per new root $u$ that wins at least one link, in addition to the increase in potential of a child of $u$ when $u$ first wins a link.  The latter is less than $\lg \size(u)$.  We consider new roots that win at least one link from first to last in link order.  We maintain the invariant that the next root $u$ that wins a link has at least $2\lg \mass(u)$ potential.  To establish the invariant for the first such root $u$, we give it $\lg \mass(u) < \lg \size(x)$ of the $2+\lg \size(x)$ additional potential.

Let $u$ be the current root under consideration and let $v$ be the next new root after $u$ that wins a link.  Node $u$ currently has potential at least $2\lg \mass(u)$.  Since $\mass(u) \geq \size(u) +\mass(v)$, the inequality $2\lg(\size(u) +\mass(v)) \geq \lg \size(u) + \lg \mass(v) + 2$ allows us to move $\lg \size(u)$ potential from $u$ to its child whose potential increases when $u$ first wins a link, and to move $\lg \mass(u)$ potential from $u$ to $v$, establishing the invariant for $v$, with at least two units of potential left over to pay for the link or links that $u$ wins during the \emph{delete-min}.  This argument covers all but the last new root $w$ in link order that wins at least one link, but $w$ ends up with potential at least $2\lg \mass(w)$, more than enough to cover the increase in potential of one of its children when it first wins a link. We use the two units of remaining added potential to pay for the actual time of the one or two links won by $w$.

The single root remaining after all the linking has a potential of  $2 + 2\lg n$.  We conclude that the amortized time of a \emph{delete-min} is at most $5 + 3\lg n$ in a slim heap, at most $5+ 4\lg n$ in a smooth heap: the links are covered by potential decreases, we added $2+\lg n$ or $2 +2\lg n$ extra potential, the remaining root needs potential $2 +2\lg n$, and the \emph{delete-min} takes time $1$ in addition to the links.  This finishes the proof that \emph{delete-min} takes $O(\lg n)$ amortized time.

\section{Implementation of \emph{decrease-key} and \emph{delete}}
\label{sec:decrease-key}

In this section we implement \emph{decrease-key} and \emph{delete}.  We give two implementations of \emph{decrease-key}, a simple one for which it is easy to prove an $O(\lg n)$ amortized bound, and a more complicated one based on Elmasry's for pairing heaps, which has an $O(\lg\lg{n})$ bound.

Our implementations of \emph{decrease-key} and \emph{delete} require adding a third pointer to each node.  If $x$ is a child that is not leftmost on its list of siblings, its third pointer indicates its previous sibling; if it is leftmost, its third pointer indicates its parent.  If $x$ is a root, its third pointer indicates the left adjacent root on the root list, or the rightmost root on the root list if $x$ is the min-root.  This makes the lists of children and the list of roots doubly linked, and supports deletion of a node (and its subtree) from such a list in $O(1)$ time.

\subsection{Key decrease}~\label{decrease-key}
A simple implementation of \emph{decrease-key}$(x, k, h)$ is the one used in pairing heaps: Set the key of $x$ to $k$.  If $x$ is not a root, cut $x$ (and its subtree) from the list of children containing it, and add $x$ to the root list.  Update the min-root.  This takes $O(1)$ actual time.  Cutting $x$ from its parent and making it a root increases the potential only of $x$, by at most $2 + 2\lg \size(x)$.  Thus the amortized time of \emph{decrease-key} is $O(\lg n)$.
A much more complicated argument extending Pettie's analysis of the same implementation of \emph{decrease-key} in pairing heaps gives an amortized bound of $O(\lg\lg n)$~\cite{STdecrease}.

We obtain more directly an amortized bound of $O(\lg\lg n)$ for \emph{decrease-key} if we use Elmasry's implementation for pairing heaps~\cite{Elmasry17}.  We maintain a \emph{buffer} that contains roots whose keys have been decreased.  The buffer of a new heap is initially empty.  To do \emph{decrease-key}($x, k, h$), we set the key of $x$ to $k$.  If $x$ is not a root, we first cut $x$ (and its subtree) from its sibling list $L$.  Second, if $x$ has a child, we cut its leftmost child $y$ (and its subtree) from $x$, and replace $x$ by $y$ in $L$.  Third, we add $x$ to the buffer.  If the buffer contains at least $\lg n$ roots, we empty it.  Whether or not $x$ was a root, we finish the \emph{decrease-key} by updating the min-root. 

When doing a \emph{delete-min}, we begin by emptying the buffer.  When doing a \emph{meld}, we empty the buffer of the smaller heap; the buffer of the larger heap becomes the buffer of the melded heap.

To empty the buffer, we sort the roots in the buffer in non-increasing order by key, and link them using leftmost locally maximum linking, which makes each root the leftmost child of the one of next-smaller key, whether the heap is a slim heap or a smooth heap. Each link is a left link.  We add to the root list the  root remaining after the roots in the buffer are linked.

To support \emph{find-min} in $O(1)$-time, we maintain the min-root of the buffer as well as the min-root of the roots not in the buffer.  This adds $O(1)$ time to each operation.  We also need to store with each heap the number of nodes it contains.  This is easy to maintain in $O(1)$ time per operation.

This way of doing \emph{decrease-key} reduces the amortized time to $O(\lg\lg{n})$.  To prove this we use the potential function of Section~\ref{sec-dfn-potential}, with three changes:

(i) If $x$ is a root in the buffer that lost a child $y$ when its key was decreased, we give its children the potential they had before $y$ was cut from $x$.  Since $y$ is the leftmost child of $x$ before it is cut, if the heap is a slim heap $y$ has zero potential. If the heap is a smooth heap, $y$ has zero potential unless $y$ and all the children of $x$ are right children, not left children.  In the latter case, $x$ has at most one child of zero potential both before and after $y$ is cut.  We conclude that after $y$ is cut, $x$ has at most one child of zero potential, its most recently acquired child if the heap is a slim heap; its most recently acquired right child if the heap is a smooth heap.  This  allows $x$ to win a new left link without increasing the potential of any of its children, including its new one.

(ii) To pay for emptying the buffer when it is full, we give each root $u$ in the buffer a potential of $4+\lg\lg{n}$, rather than the $2 + \lg \size(u)$ potential it would have if it were in the root list.

(iii) To pay for emptying the buffer of a heap when it is melded with a larger heap, we give each heap of $n$ nodes an extra potential of $4\lg n$.

The actual time to decrease the key of a node $x$ is $O(1)$.  After $x$ is cut from its sibling list and its leftmost child is cut from it, its children have exactly the potential they need when $x$ is in the buffer.  Replacing $x$ by the first child of $x$ does not increase the potential of any node in the tree previously containing $x$.  Thus the potential increase caused by a \emph{decrease-key} is at most the potential of the new root in the buffer, which is $4+\lg\lg{n}$, making the amortized time of the \emph{decrease-key} $O(\lg\lg{n})$ unless the buffer becomes full and is emptied.

The actual time to empty the buffer is $O(\lg\lg{n})$ per root in the buffer, with the sorting time dominant.  The potential of $\lg\lg{n}$ per root in the buffer covers this time.  Because of the sorting, each root in the buffer acquires at most one new left child, which does not increase the potential of any node in any of the trees rooted at these nodes.  One root, say $u$, remains after the roots in the buffer are linked.  This node needs potential $2+2\size(u) \leq 4\lg n$ when it is added to the root list. Since the buffer is full when it is emptied in a \emph{decrease-key}, $u$ acquires the needed potential from the extra $4$ units of potential per root in the buffer before these roots are linked.  Hence \emph{decrease-key} takes $O(\lg\lg{n})$ time whether or not the buffer is emptied.

If the buffer is emptied during a \emph{delete-min}, the root formed by linking the roots in the buffer needs potential at most $2+ 2\lg n$, increasing the amortized time of the \emph{delete-min} by this amount, but the amortized time remains $O(\lg n)$.

It remains to consider \emph{insert} and \emph{meld}.  Insertion is just a special case of meld, in which one of the heaps contains only one item and an empty buffer.  Consider a meld of two heaps, with the larger heap, say $h$, containing $n$ nodes.  The heap resulting from the meld has size at most $2n$. The heap potential of the smaller heap covers the potential of the root formed by linking the roots in its buffer when it is emptied.  The heap potential of the new heap is at most $4\lg(2n) = 4\lg n +4$, which is an increase of at most $4$ of the heap potential of $h$.  Let $k \leq \lg n$ be the number of roots in the buffer of $h$.  The meld increases the sum of the potentials of these roots by at most 
\[k(\lg\lg{(2n)}-\lg\lg{n}) \leq (\lg n)\lg\left(\frac{\lg n +1}{\lg n}\right) = (\lg n)\lg(1 +1/\lg n) \leq (\lg n)(\lg e)/\lg n = \lg e\]
by the inequality $\lg(1 + 1/\lg n) \leq \lg e/\lg n$ if $n \geq 2$.  We conclude that the amortized times of \emph{insert} and \emph{meld} remain $O(1)$.

\subsection{Arbitrary deletion}

To delete a node $x$ in a heap $h$, if $x$ is the min-root of $h$, we merely do \emph{delete-min}($h$).  If $x$ is not the min-root of $h$, we offer three ways of doing the deletion.  One is to decrease the key of $x$ to minus infinity and then do \emph{delete-min}($h$).  Using either of our implementations of \emph{decrease-key}, this takes $O(\lg n)$ amortized time, with the \emph{delete-min} time dominating.  An alternative that does not use \emph{decrease-key} is to repeatedly link adjacent children of $x$ using leftmost locally maximum linking until $x$ has only one child, and then replace $x$ by its only child in the root list, or in its list of siblings if $x$ is not a root.  This implementation of \emph{delete} also takes $O(\lg n)$ time, by an analysis like that of \emph{delete-min} in Section~\ref{bounds}.  A third alternative is to delay linking the children of the deleted node, by merely replacing the deleted node in its list of siblings (or in the root list if it is a root) by its list of children.  This also takes $O(\lg n)$ amortized time, by an extension of the analysis in Section~\ref{bounds}: we must add $O(\lg n)$ additional potential to each group of siblings whose parent is deleted, to help pay for the links they eventually win. 

\section{Permanent nodes do not count}
\label{sec:remain}

In this section we prove Theorem~\ref{thm:remain}, which states that nodes that are never deleted do not slow down operations:

\restatethma*
\begin{proof}
Call a node \textit{temporary} if it will eventually be deleted, and \textit{permanent} otherwise.
The total time of a sequence of operations is O(1) per operation plus the number of links.  All the links except at most one per \emph{decrease-key} operation are done during \emph{delete-min} operations and when a buffer of nodes whose keys have decreased is emptied.  Thus it suffices to count links.  We separately count links won by permanent nodes and those won by temporary nodes.  A link won by a permanent node either remains permanent, or is cut by a subsequent \emph{decrease-key operation}: such links cannot be cut by deletions.  It follows that there is at most one such link per \emph{insert} and at most two per \emph{decrease-key}.

It remains to count links won by temporary nodes.  To do this we redefine the size of a node to be the number of temporary nodes in its subtree, and we redefine mass and potential accordingly, except that we give a root with size $0$ potential $2$ and a child with mass $0$ potential $0$.
It is easy to check that with the new potential function \emph{make-heap}, \emph{find-min}, \emph{insert}, and \emph{meld} still take constant amortized time, and \emph{decrease-key} still takes $O(\lg\lg{n})$ amortized time.%

We make a few small changes to adapt the analysis of \emph{delete-min} to the new potential function.  We apply the analysis of links won by old roots only to those that are temporary.  We add and shift potentials only among temporary new roots.  The additional potential needed for the analysis to work is $O(\lg t)$ rather than $O(\lg n)$.  Thus we obtain an $O(\lg t)$ amortized time bound for \emph{delete-min}.  We also obtain an $O(\lg t)$ amortized time bound for either implementation of arbitrary deletion that does not use \emph{decrease-key}; for the implementation that uses \emph{decrease-key}, the bound is $O(\lg t + \lg\lg{n})$. 
\end{proof}

\section{Experimental evaluation}\label{sec:exp}

We implemented the pairing heap, the smooth heap, and the slim heap in %
Python 3, 
with 
the same generic heap interface.  We compared their performance in two scenarios: sorting and Dijkstra's single-source shortest paths algorithm. %
In both tasks %
we counted two types of logical operations: \emph{comparisons} between keys and \emph{links} between nodes. Note that in pairing heaps comparisons and links occur together, making the two counts equal. By contrast, in smooth heaps and slim heaps the number of comparisons is larger than the number of links by a factor of at most two (see~\cite{KS19})\footnote{In the worst case $2k - \Theta(\lg{k})$ comparisons are necessary to combine $k$ roots during \emph{delete-min}, since the outcome corresponds to one of $\Theta(4^k/k^{3/2})$ possible binary trees.}. In practice, links are typically costlier than comparisons, requiring several pointer changes. Thus the number of links is expected to be indicative of the actual running time. %
We summarize the experimental setup and findings (Figures~\ref{fig:exp1}, \ref{fig:exp2}, \ref{fig:exp3}) and refer to \cite{Hartmann_thesis} for further experiments.\footnote{The scripts used can be found at \url{https://git.imp.fu-berlin.de/hlm/smooth-heap-pub}.} 

\subparagraph*{Sorting.} 
Given a list of roots containing the keys $\pi_1, \dots, \pi_n$ of an input permutation $\pi$, we execute $n$ \emph{delete-min} operations, sorting $\pi$. 
It has been hypothesized that the smooth heap performs particularly well on structured inputs. To test this experimentally, we considered, besides uniformly random inputs, classes of \textit{semi-random} permutations. We tested four families of input permutations: %

\emph{(a) Uniformly random}: We generated permutations of sizes $n = 2,2^2,\dots,2^{17}$ using the pseudo-random \emph{random.shuffle} function of Python, with $5$ independent runs for each size. 

\emph{(b) Separable:} Starting with the sequence $1, \dots, n$, for $n = 2,2^2,\dots,2^{17}$, we did the  following shuffling: reverse the sequence with probability ${1}/{2}$, then recursively process the first half and second half of the sequence in the same way, doing $20$ independent runs for each size. The permutations obtained are \emph{separable} (see e.g.~\cite{BBL}). %
\emph{(c) Localized:} We generated a sequence of length $n = 10000$ where each element is drawn from a Gaussian normal distribution centered at its index, with standard deviation proportional to a parameter $\varepsilon$, doing $10$ independent runs for each value $\varepsilon = 0, 0.01, \dots, 0.3$. %

\emph{(d) Sorted blocks:} Starting with a uniform random permutation of size $n=10000$, we sorted contiguous blocks of elements, where the block sizes are chosen uniformly at random from the range $[1, \dots B]$, with $20$ independent runs for each value $B = 100,200,\dots,2000$.

\subparagraph*{Shortest paths.}

For all three heaps we used a similar multi-tree implementation of \emph{insert} and \emph{decrease-key} that simply adds the respective node to the end of the root list.  (This is \emph{not} the standard pairing heap implementation, which maintains a single tree, and the implementation of \emph{decrease-key} in smooth heaps and pairing heaps is the simple one.) Since these structures differ only in the restructuring during delete-min, we counted the number of links and comparisons only for delete-mins, taking averages over $10$ independent runs. We considered two families of undirected input graphs, with edge weights in both cases chosen uniformly at random from $\{1,\dots,10000\}$:

\emph{(e) Random graphs: } with $n = 500$ vertices, generated according to the Erd\H{o}s-R\'enyi model, with edges present with probability $p$ (for values $p = 0,0.05,\dots,1$), and 

\emph{(f) Random $10$-regular graphs:} with $n=500,1000,\dots,10000$ vertices. 
\subparagraph*{Findings.}

In (a) the number of comparisons (and links) done by the pairing heap is a small factor ($\approx 1.2$) above the information-theoretical lower bound of $\lg(n!)$. On our largest test, the smooth (slim) heap performs about $22\%$ ($44\%$) more comparisons. In turn, the pairing heap performs about $47\%$ ($11\%$) more links than the smooth (slim) heap.

In (b) the number of operations appears to be $O(n)$ for all data structures. This is consistent with the fact that there are $2^{O(n)}$ such permutations of size $n$. In fact, for the smooth heap, it is known~\cite{KS19} that sorting separable permutations takes time $O(n)$, as implied by more general results for Greedy BST. It can be observed that the pairing heap performs on average both more comparisons and more links than the smooth and slim heaps. The differences in the number of comparisons are small, but in the number of links the overhead of pairing heaps is about $58\%$ ($65\%$) in our largest test. One can further notice that there is less variation in the number of operations in smooth/slim heaps, whereas the pairing heap is sensitive to the exact choice of input permutation. 

In (c) the situation is similar to the uniformly random case, with the ordering slim, {smooth}, pairing in the number of comparisons, and pairing, slim, smooth in the number of links. Although the overall costs are smaller to (a), the smooth/slim heaps do not appear to have a particular advantage in adapting to this type of structure.   

In (d) the smooth/slim heaps have a clearer advantage. When the input permutation consists of just about $10$ sorted blocks, the pairing heap performs about $7$ comparisons and links per key, whereas both  smooth and slim heaps perform about $6$ comparisons and $4$ links per key. The advantage diminishes as the sizes of the sorted blocks decrease.  

In the densest case of (e), the smooth (slim) heaps perform about $40\%$ ($45\%$) more comparisons than the pairing heap, while the pairing heap performs about $22\%$ ($14\%$) more links than the smooth (slim) heaps. In the largest case of (f), the smooth (slim) heaps perform about $36\%$ ($48\%$) more comparisons, while the pairing heap performs about $23\%$ ($6\%$) more links.

\begin{figure}[ht!]
\begin{center}
\begin{subfigure}[t]{.49\textwidth}
  \centering{\includegraphics[height=2.3in]{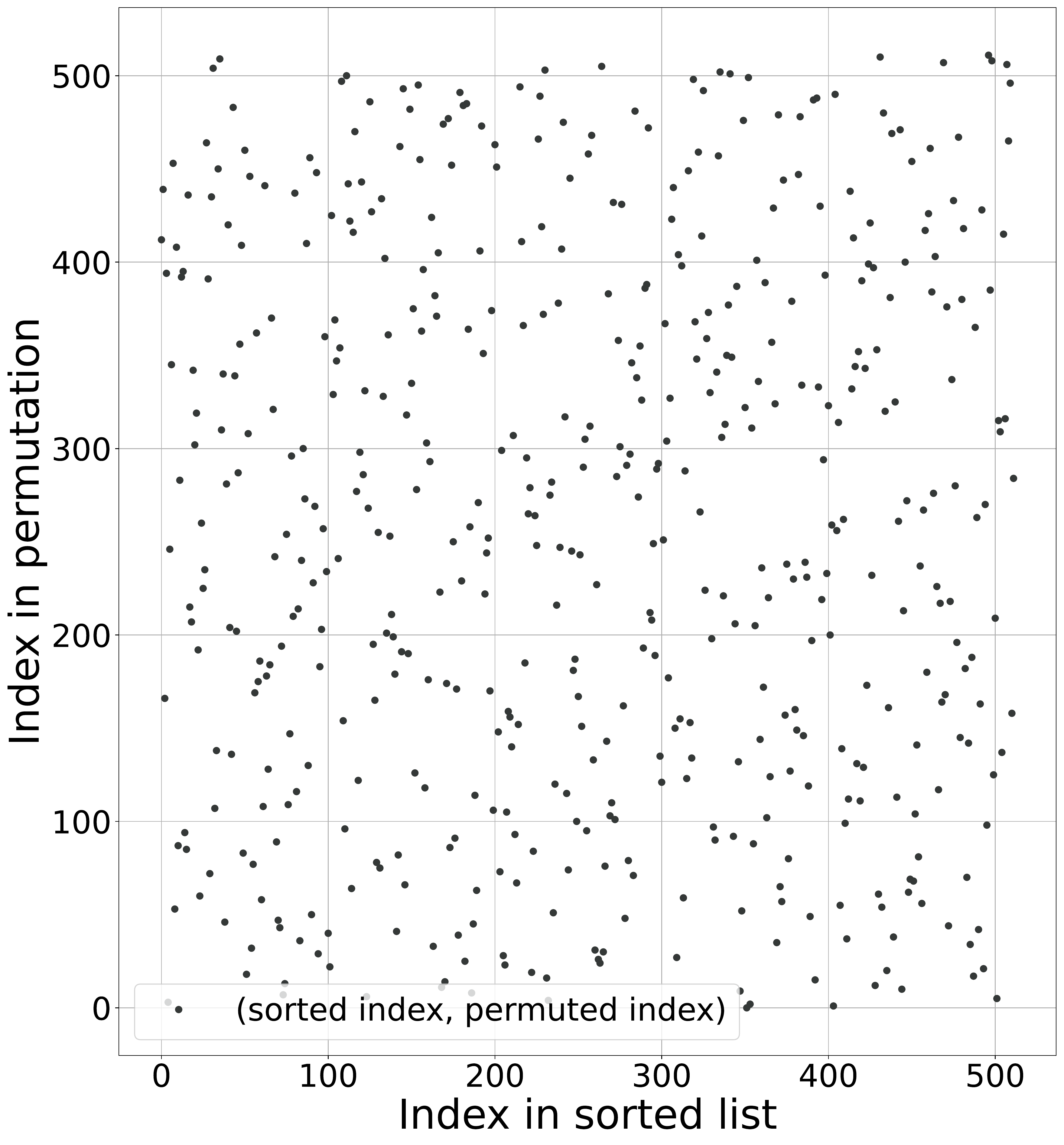}}
  \caption{Uniform random permutation \label{subfig:uniformly-random-input}}
\end{subfigure}
\begin{subfigure}[t]{.49\textwidth}
  \centering{\includegraphics[height=2.3in]{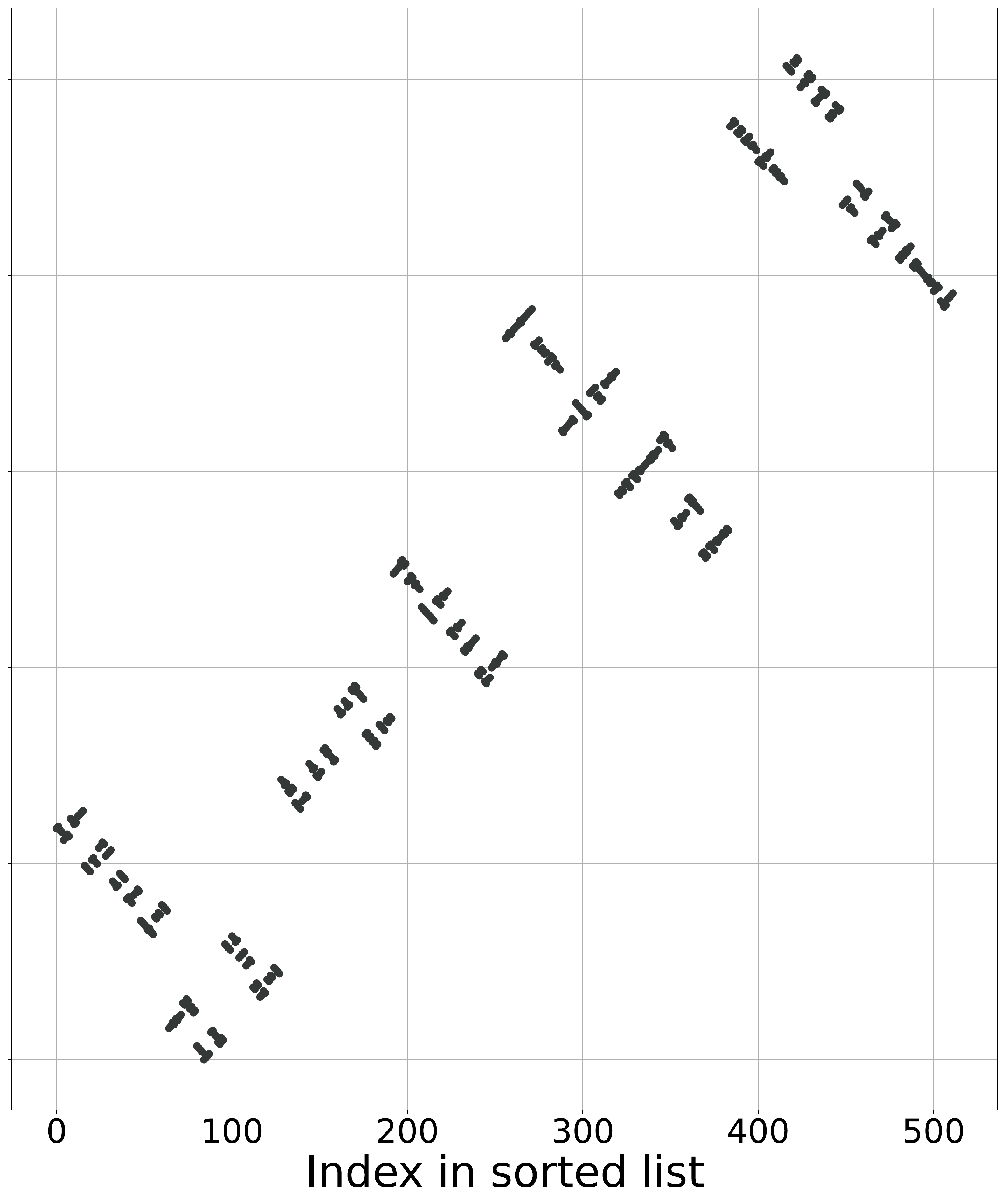}}
  \caption{Random separable permutation \label{subfig:separable-permutation-input}}
\end{subfigure}

\vspace{0.28in}
\begin{subfigure}[t]{.49\textwidth}
\centering{
  \hspace{0.27in}\includegraphics[height=2.3in]{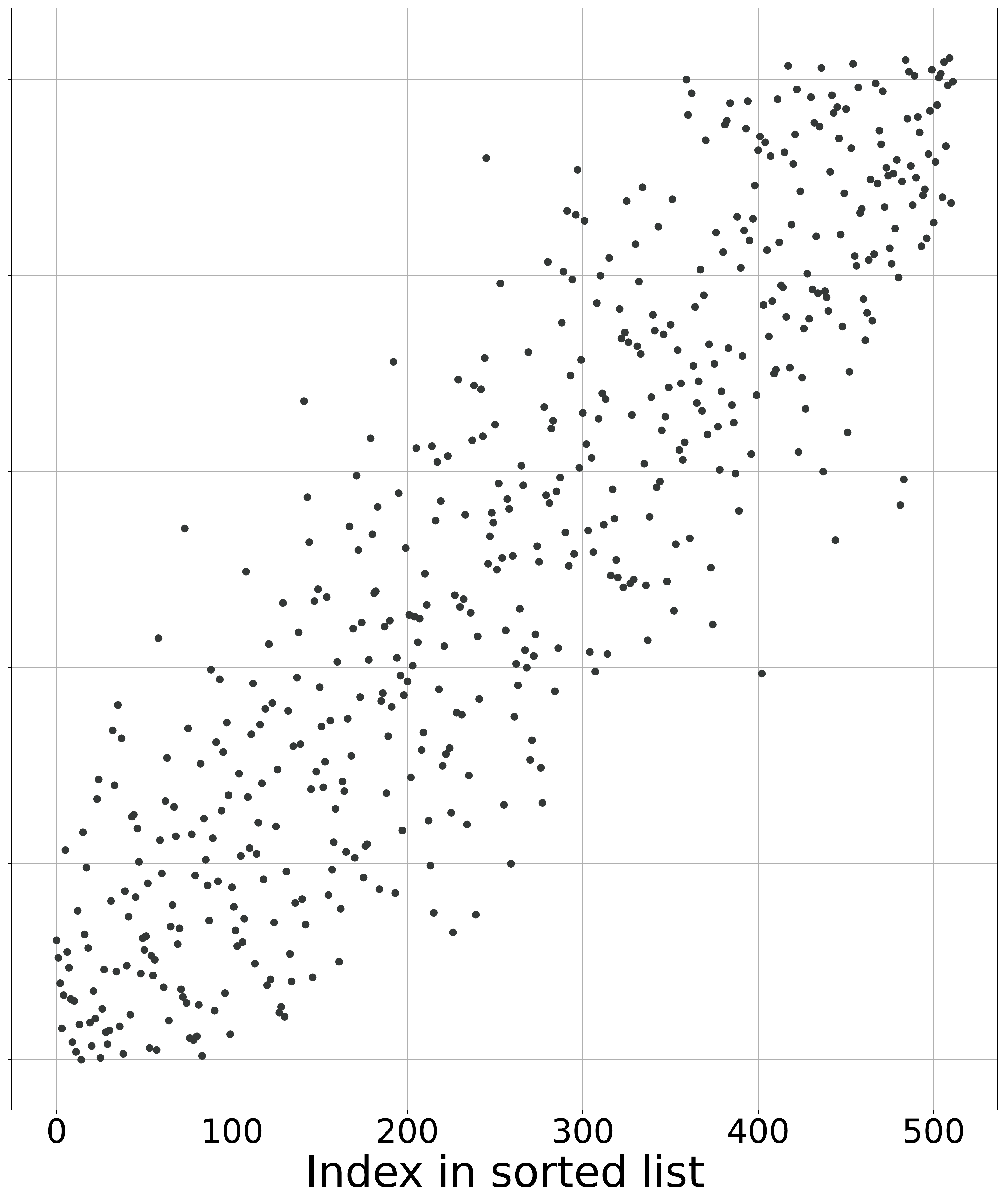}
  }
\caption{Random permutation with locality parameter $\varepsilon = 0.15$ \label{subfig:localised-permutation-input}}
\end{subfigure}
\begin{subfigure}[t]{.49\textwidth}
  \centering{\includegraphics[height=2.3in]{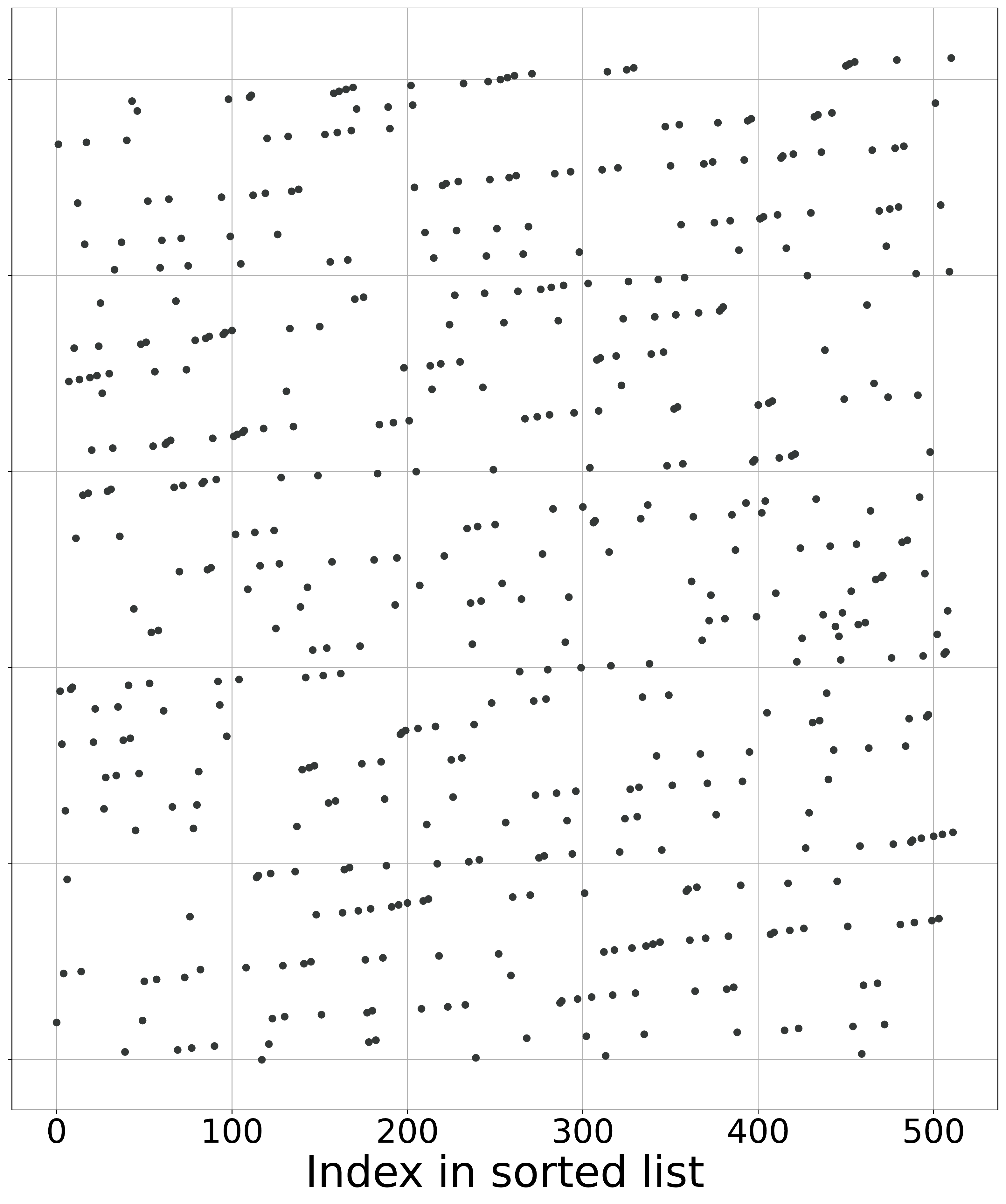}}
  \caption{Random permutation with average sorted blocks of length $15$\label{subfig:subseq-permutation-input}}
\end{subfigure}
\end{center}
\caption{Sample permutations of size $512$ generated for the different families of permutations.\label{fig:exp1}}
\end{figure}

\begin{figure}[ht!]
\centering{
\begin{subfigure}[t]{.49\textwidth}
\centering{\includegraphics[height=2.3in]{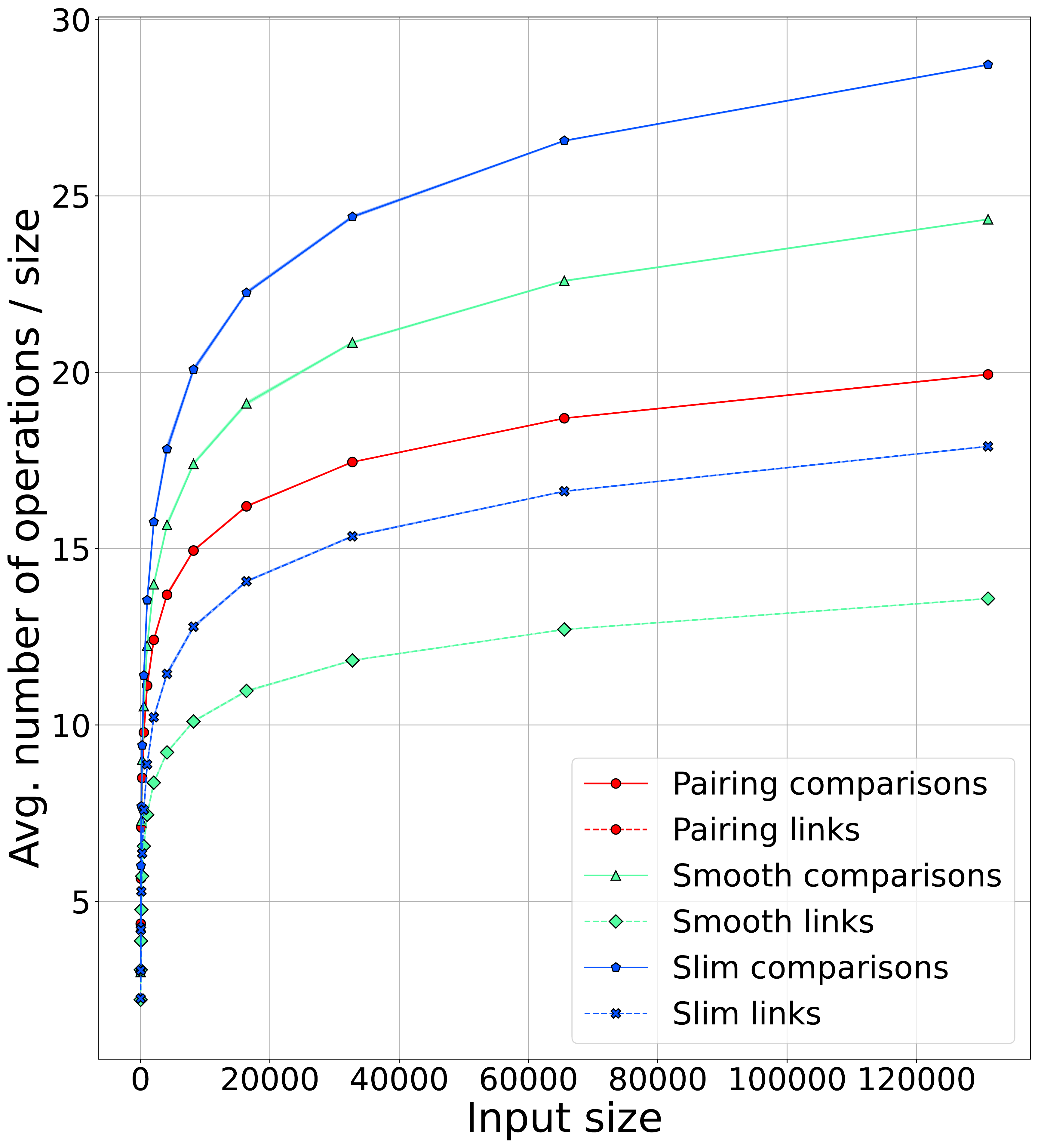}}
\caption{Uniform random permutations \label{subfig:uniformly-random-sorting}}
\end{subfigure}
\begin{subfigure}[t]{.49\textwidth}
  \centering{\includegraphics[height=2.3in]{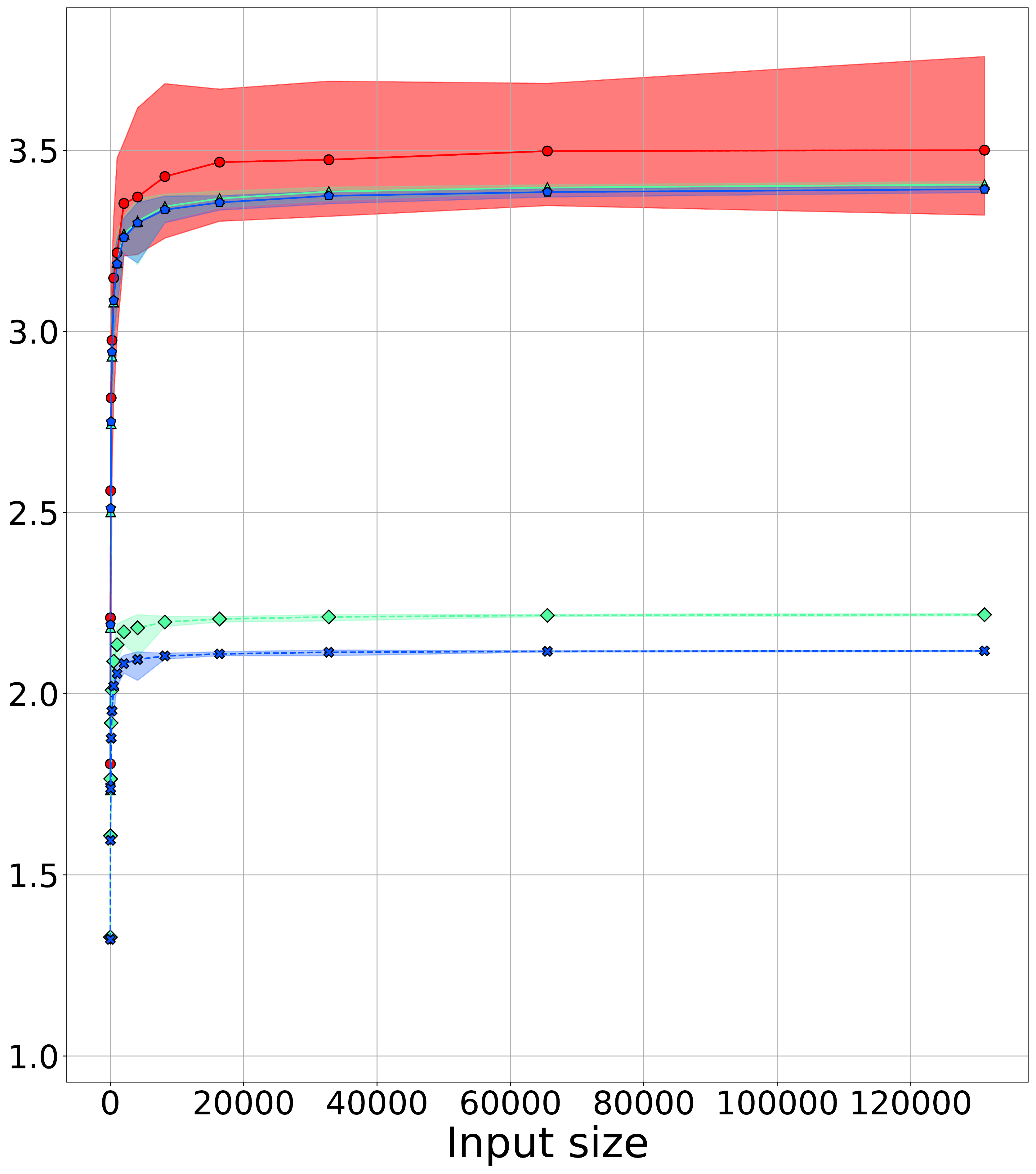}}
  \caption{Random separable permutations \label{subfig:separable-permutation-sorting}}
\end{subfigure}
\vspace{0.28in}

\begin{subfigure}[b]{.49\textwidth}
  \centering{\includegraphics[height=2.3in]{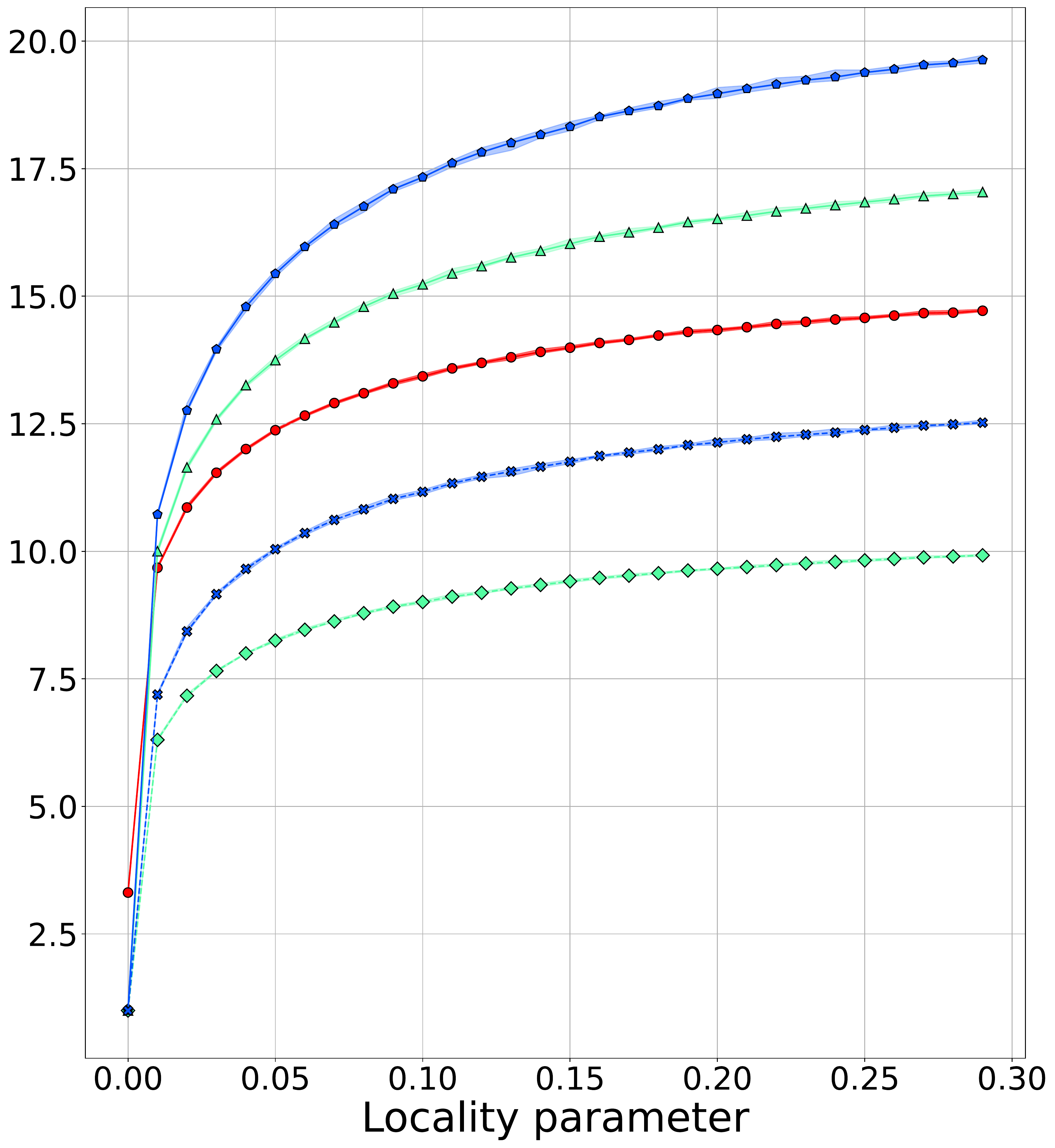}}
  \caption{Random permutations of size $10^4$ with locality parameter $\varepsilon$ \label{subfig:localised-permutation-sorting}}
\end{subfigure}
\begin{subfigure}[b]{.49\textwidth}
  \centering{\includegraphics[height=2.3in]{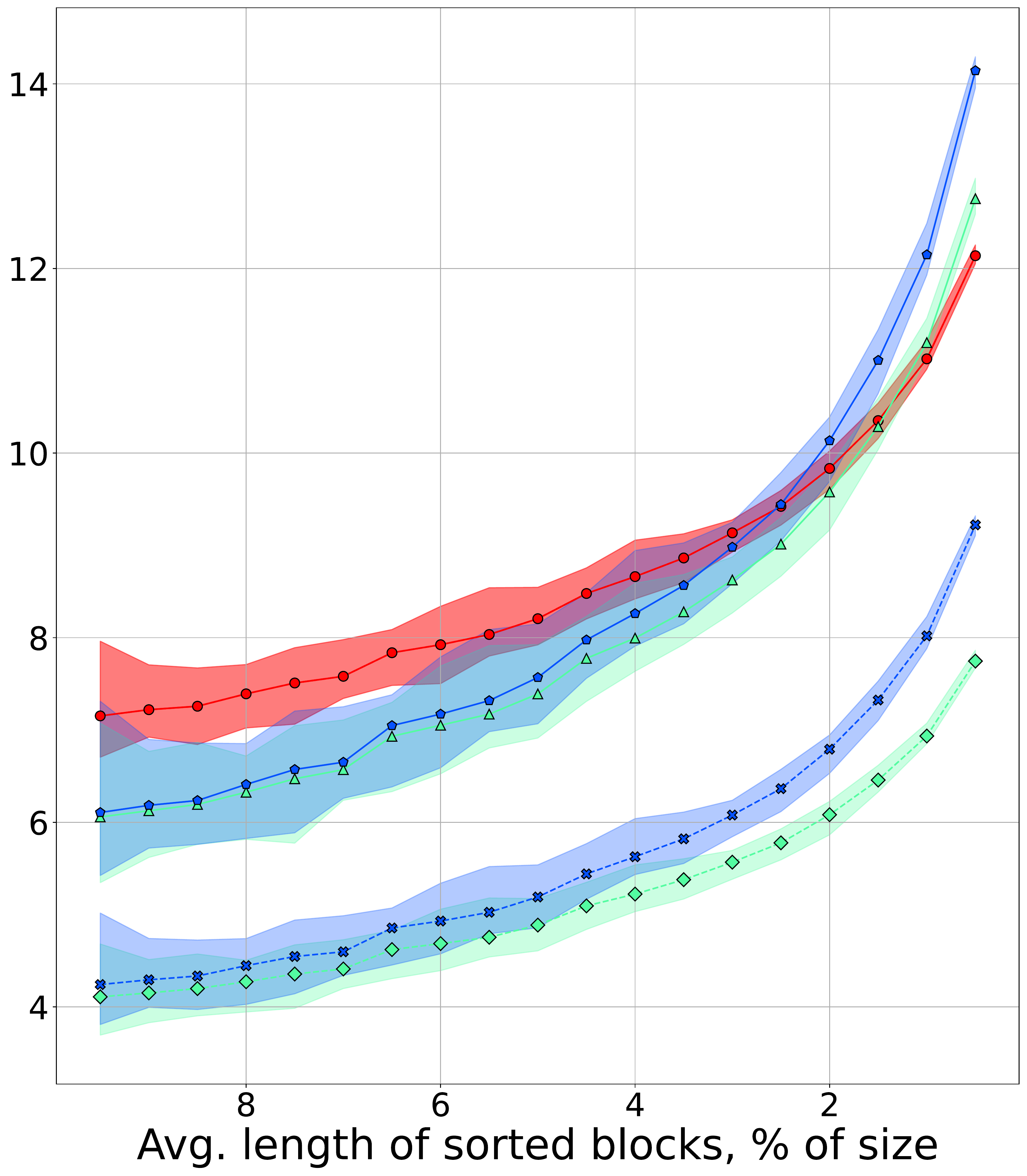}}
\caption{Random permutations of size $10^4$ with sorted blocks}
\end{subfigure}
}
\caption{Average number of links (dashed lines) and average number of comparisons (continuous lines) for sorting permutations using pairing heap (red), smooth heap (green), and slim heap (blue); note that for pairing heaps the two counts are equal. Shaded areas indicate ranges between minimum and maximum costs.\label{fig:exp2}
}
\end{figure}

\begin{figure}[ht!]
\centering{
\begin{subfigure}[b]{.49\textwidth}
  \centering{\includegraphics[height=2.3in]{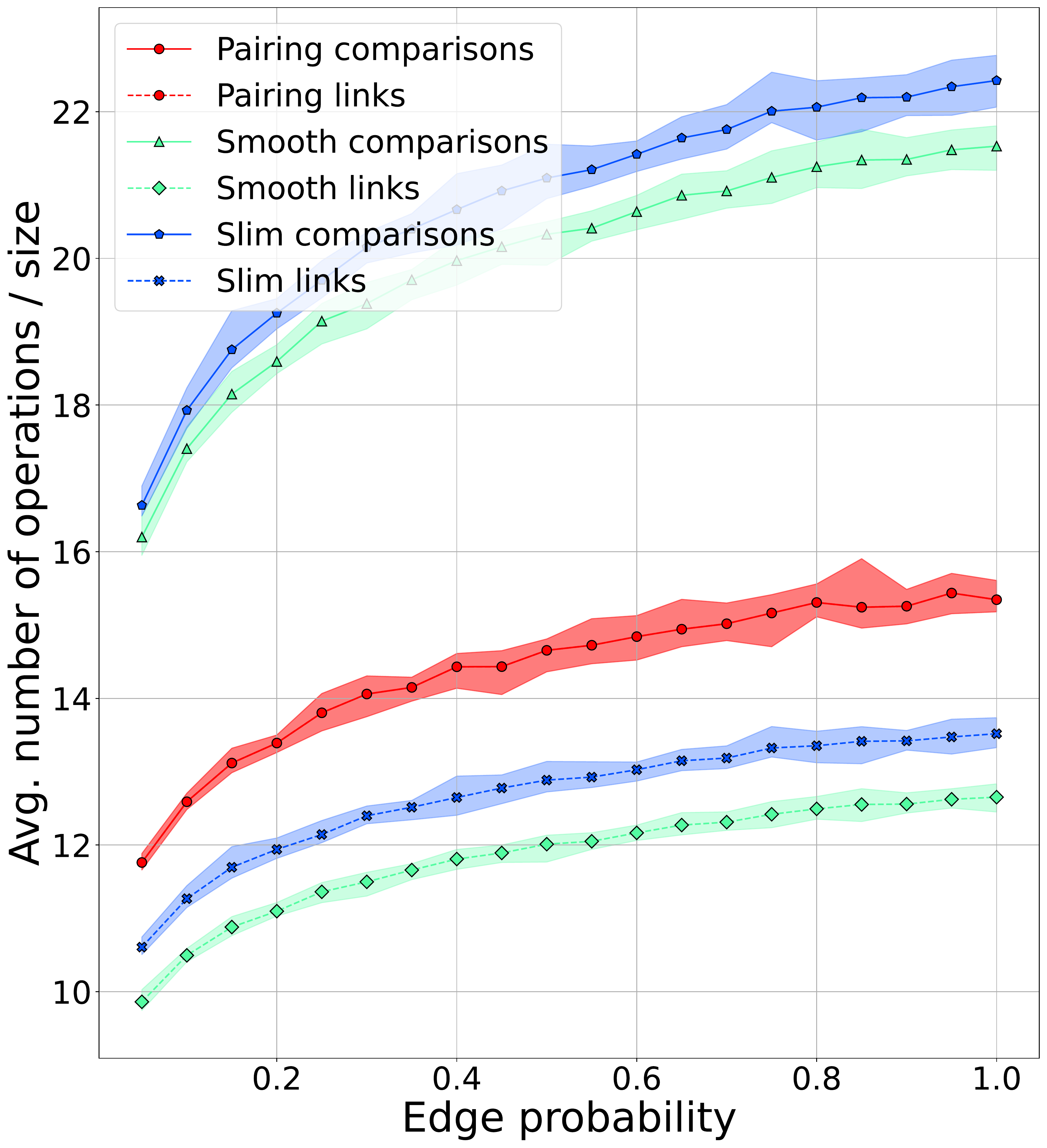}}
\caption{Random graphs with varying edge probability.}
\label{fig:exp-dijkstra1}
\end{subfigure}
\begin{subfigure}[b]{.49\textwidth}
  \centering{\includegraphics[height=2.3in]{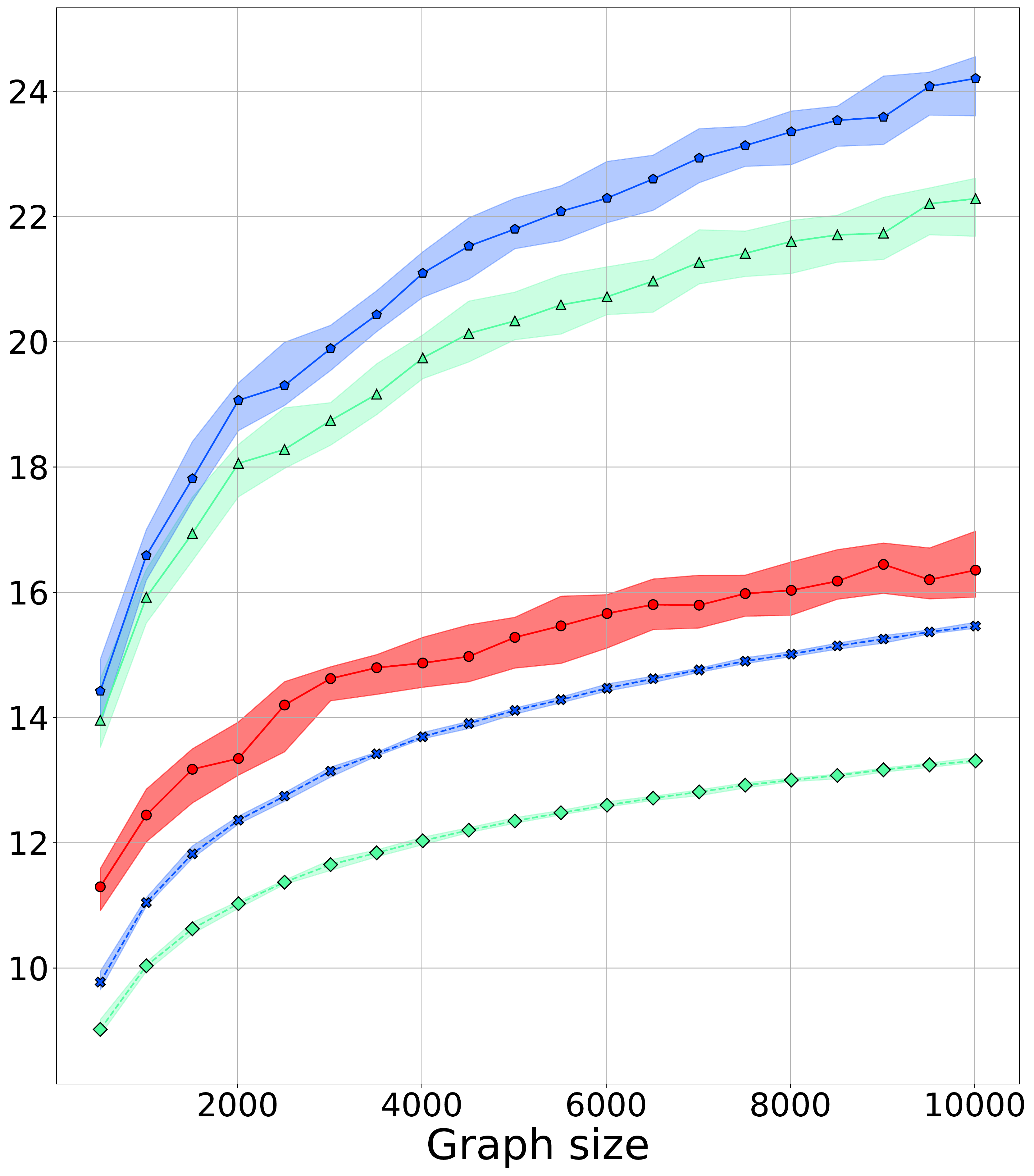}}
  \caption{Random regular graphs of varying size.}
\end{subfigure}
\caption{Average number of links (dashed lines) and average number of comparisons (continuous lines) for Dijkstra's algorithm using pairing heap (red), smooth heap (green), and slim heap (blue).\label{fig:exp3}}
\label{fig:exp-dijkstra2}
}
\end{figure}

\section{Remarks}

Self-adjusting data structures are simple to describe and implement but hard to analyze. Developing new approaches and tools for analyzing such data structures remains an exciting field, with many questions still open. Our design of slim heaps and our analysis of smooth and slim heaps (Theorems~\ref{thm:unweighted} and \ref{thm:remain}) add to what we know about such structures. 

One property that makes smooth heaps and slim heaps efficient is that each \emph{delete-min} links at most two new children to each node.  Neither the original version of pairing heaps nor any of the proposed variants has this property, but there is one variant that does: the \emph{pure pairing heap}, in which the heap is a forest instead of a single tree, \emph{insert} and \emph{meld} are done lazily by catenating the root lists, and \emph{delete-min} performs a single pairing pass: after the min-root is deleted, the remaining roots are linked in pairs, first and second, third and fourth, and so on.  During a \emph{delete-min}, each node acquires at most \emph{one} new child.  We offer the pure pairing heap as a data structure for further study: our analysis fails, because this structure does not have the second property that makes smooth and slim heaps efficient; a \emph{delete-min} can result in a list of roots rather than a single one.

We have proved that the simple method of doing \emph{decrease-key} in smooth and slim heaps, by merely cutting the node whose key decreases and adding it to the root list, takes $O(\lg n)$ amortized time.  For pairing heaps, Pettie has shown that the same method takes $O(4^{\sqrt{\lg\lg{n}}})$ amortized time.  Two authors of this paper~\cite{STdecrease} have recently adapted and extended Pettie's analysis to show that the simple decrease-key implementation in fact takes $\Omega(\lg\lg n)$ amortized time in smooth and slim heaps.

Similarly to other self-adjusting data structures, smooth heaps and slim heaps are expected to be \emph{adaptive}, i.e.\ to show greater efficiency on some specific inputs. Through the connection with Greedy BST~\cite{KS19}, smooth heaps are known to be highly adaptive in sorting mode, but these results do not seem easy to transfer to slim heaps. Adaptivity in general operation sequences is much less understood. Our Theorem~\ref{thm:remain} can be seen as a first result in this direction. We expect that other adaptive bounds, e.g.\ similar to those in~\cite{EFI12} can also be shown. A possible way for proving such bounds is to generalize the analysis in this paper to a weighted setting, much in the spirit of splay trees~\cite{ST85}. Such an extension of Theorem~\ref{thm:unweighted} is indeed possible, but with the current potential function it does not appear to yield nontrivial new bounds. A weighted analysis combined with a \emph{linear} potential function, i.e.\ $O(n)$ rather than $O(n\lg n)$ may lead to such results.

\clearpage
\bibliography{SH}
\end{document}